\documentclass[11pt,reqno]{amsart}
\usepackage[top=2.54cm, bottom=2.54cm, left=2.8cm, right=2.8cm]{
geometry}
\usepackage{dsfont, amssymb,amsmath,amscd,latexsym, amsthm, amsxtra,amsfonts}
\usepackage[all]{xy}
\usepackage[active]{srcltx}

\usepackage[round]{natbib}
\usepackage{bbm}
\usepackage{enumerate}
\usepackage{mathrsfs}
\usepackage{graphicx}
\usepackage{comment}
\usepackage{mathtools}

\usepackage{tikz}
\usetikzlibrary{calc,arrows}
\usepackage{verbatim}
\usepackage{graphicx}
\usepackage{subfigure}

\newtheorem{theorem}{Theorem}[section]

\newtheorem{proposition}[theorem]{Proposition}
\newtheorem{lemma}[theorem]{Lemma}

\theoremstyle{definition}
\newtheorem{definition}[theorem]{Definition}

\renewcommand{\theequation}{\arabic{section}.\arabic{equation}}
\newtheorem{Def}{Definition}[section]

\theoremstyle{definition}

\theoremstyle{definition}
\newtheorem{remark}{Remark}
\theoremstyle{definition}

\newcommand{\rd}{\mathrm{d}}

\renewcommand{\epsilon}{\varepsilon}

\newcommand{\E}{\mathbf{E}}

\usepackage[pdfstartview=FitH, bookmarksnumbered=true,bookmarksopen=true, colorlinks=true, pdfborder=001, citecolor=blue, linkcolor=blue,urlcolor=blue]{hyperref}
\usepackage{graphics}
\graphicspath{{figures/}}
\begin{document}
\makeatletter
\def\@setauthors{%
\begingroup
\def\thanks{\protect\thanks@warning}%
\trivlist \centering\footnotesize \@topsep30\p@\relax
\advance\@topsep by -\baselineskip
\item\relax
\author@andify\authors
\def\\{\protect\linebreak}%
{\authors}%
\ifx\@empty\contribs \else ,\penalty-3 \space \@setcontribs
\@closetoccontribs \fi
\endtrivlist
\endgroup} \makeatother
 \baselineskip 19pt
 \title[{{\tiny Robust equilibrium game in a DB pension plan}}]
 {{\tiny
Robust equilibrium strategies in a defined benefit pension plan game}} \vskip 10pt\noindent
\author[{\tiny  Guohui Guan, Jiaqi Hu, Zongxia Liang}]
{\tiny {\tiny  Guohui Guan$^{a,*}$, Jiaqi Hu$^{b,\dag}$, Zongxia Liang$^{b,\ddag}$}
 \vskip 10pt\noindent
{\tiny ${}^a$School of Statistics, Renmin University of China, Beijing 100872, China
\vskip 10pt\noindent\tiny ${}^b$Department of Mathematical Sciences, Tsinghua
University, Beijing 100084, China
}\noindent
\footnote{{\tiny $^*$ {\bf e-mail}: guangh@ruc.edu.cn\\
$^\dag $ Corresponding author, \ \ {\bf e-mail}: hujq20@mails.tsinghua.edu.cn      \\
 $\ddag $ {\bf e-mail}: liangzongxia@mail.tsinghua.edu.cn}}}
\numberwithin{equation}{section}
\maketitle
\noindent
\begin{abstract}
This paper investigates the robust {non-zero-sum} games in an aggregated {overfunded} defined benefit (abbr. DB) pension  plan. The sponsoring firm is concerned with the investment performance of the fund surplus while the participants act as a union to claim a share of the fund surplus. The financial market consists of one risk-free asset and $n$ risky assets. The firm and the union both are ambiguous about the financial market and care about the robust strategies under the worst case scenario. {The union's objective is to maximize the expected discounted utility of the additional benefits, the firm's two different objectives are to maximizing the expected discounted utility of the fund surplus and  the probability of the fund surplus  reaching  an upper level before hitting a lower level in the worst case scenario.} We formulate the related two robust non-zero-sum games for the firm and the union. Explicit forms and optimality of the solutions are shown by stochastic dynamic programming method. In the end of this paper, numerical results are illustrated to depict the economic behaviours of the robust equilibrium strategies in these two different games.
\vskip 10 pt \noindent
{JEL classification: C61,G11, G22, D53.
\vskip 5pt  \noindent
Submission Classification: IE11, IE13, IM50, IB81.}
\vskip 5pt \noindent
Keywords: {Overfunded} DB pension plan; Robust control; Stochastic differential game; Nash equilibrium; Stochastic dynamic programming.
\end{abstract}
\vskip10pt
\setcounter{equation}{0}

\section{{{\bf Introduction}}}
With the aging of population worldwide, retirement plan plays an important role in ensuring the quality of life of the elders. Although there is a recent shift from {the DB pension plan to the defined contribution (abbr. DC) pension plan,} DB pension plan is still very important in retirement welfare policy. The benefits of the participants in DB pension plan are fixed in advance and the contribution rates are calculated and adjusted to keep  actuarial balance. In a DB pension plan, benefits are distributed at retirement as  an annuity or {one lump-sum payment}, and constitute the main liability of the firm. Besides, the contributions from the participants constitute the assets of the firm. In order to hedge financial risk, the firm allocates the wealth of the fund in the financial market.
\vskip 5pt
When the asset of the firm is lower than its liability, the pension fund is underfunded. A study by  investment bank Credit Suisse First Boston (CSFB) finds that at the end of 2001, there are more companies in the S\&P 500 with underfunded DB pension plans than at any time during the previous 10 years. When the pension fund is underfunded, the firm is faced with insolvency risk and aims to minimize the gap between the liability and asset. Many previous studies in DB pension fund focus on the underfunded case, see e.g., \cite{Haberman1994Dynamic}, \cite{Haberman2000Contribution},  \cite{Josa2004Optimal}, etc. {Most of recent work} consider more financial risks for the underfunded DB pension plan and can help manage the fund better, see interest risk in  \cite{huang2006control}, \cite{Hainaut2011Optimal}, jump risk in  \cite{Josa2012Stochastic}, volatility risk in \cite{Josa2018Portfolio}, etc.
\vskip 5pt
Although most of the pension plans are underfunded, overfunding may also happen {because of} the outstanding performance of the investment, particularly in the bull market, {or  the premature death of a high paid employee}. Table II(c) in \cite{oecd} illustrates that a small proportion of companies are overfunded. \cite{Sp500} show that within the S\&P
500, 46 companies had overfunded pensions in 2018 (up from 43 in 2017). Overfunded DB pension plan has its own issues about how to manage the excess assets. Reversion when liquidated, increasing plan benefits, allowing additional accruals, adding plan participants may be possible tools to manage the excess assets. However, an excise tax 50\% is ultimately imposed on plan reversion. \cite{2014INVESTOR} {show} that firms terminating overfunded pension plans tend to have tax loss from termination.  \cite{2013Pension} indicate that the firms have a desire to modify the implicit contract instead of terminating the plan. As such, distributing benefits to plan participants is a better tool in an overfunded pension fund.  In \cite{Josa2019Equilibrium}, the authors assume that in the overfunded DB pension plan, the participants claim a share of the fund surplus as additional benefits. On the one hand, increasing plan benefits can avoid the tax from termination. On the other hand, the additional benefits can attract more potential participants.
\vskip 5pt
\cite{Josa2019Equilibrium} present games for the firm and the union. The firm cares about the wealth of the fund surplus and has two goals while the participants aim to acquire more additional benefits from the fund. The explicit solutions of the equilibrium strategies are derived in \cite{Josa2019Equilibrium}, {their results show that} the strategies rely heavily on the parameters of the financial market. The equilibrium game in \cite{guan2016stochastic} for the DC pension fund also shows the great impacts of parameters on the strategies. However, as the limited information of the financial market, the investors can not observe the financial market accurately and are in fact ambiguity averse when making decisions, see \cite{Hansen2001Robust}, \cite{Binmore2012How}, \cite{guan2018time}, \cite{guan2019robust}, etc. In the management of DC pension fund, \cite{2017RobustSL}, \cite{2018RobustPL} and \cite{2019RobustPL} implicate that when ignoring ambiguity, the manager is faced with large utility losses. As such, model uncertainty should be considered in risk management.
\vskip 5pt
In this paper, we study the robust equilibrium games between the firm and participants. The goals ignoring ambiguity are the same as in  \cite{Josa2019Equilibrium} and then modified to introduce ambiguity aversion for the players. The firm and participants seek the robust equilibrium strategies under worst case scenario. In fact, recently, there are some literature about the robust equilibrium games, see \cite{Pun2016Robust}, \cite{2018RobustWZ}, \cite{2020Reinsurance},  etc. These mentioned work concerning robust equilibrium game mainly consider two separate agents and the agents care about the relative performance over the other. Besides, both the agents have the same instruments to manage the wealth. However, in our work, the decisions made by the firm and participants affect the same wealth process. The firm and participants have different tools: the firm chooses the investment in the financial market, the participants decide the amount of additional benefits. Moreover, the previous mentioned work suppose that the agents will never go bankruptcy and  are concerned with the wealth of the agents at a fixed time. In this paper, we consider two different goals for the firm: maximizing the expected discounted utility of the fund surplus and  the probability of the fund surplus  reaching  an upper level before hitting a lower level.  On the  one hand, our goals are more abundant and  provide guidance for different kinds of agents. On the other hand, concerning the robust optimization problem at a fixed time, the procedure in \cite{Mataramvura2008Risk} can be well applied to show the optimality of the robust equilibrium strategies. The goals {in our paper} are different from the optimization rule in \cite{Mataramvura2008Risk}  and detailed verification theorem is required.
\vskip 5pt
 {Using} stochastic dynamic programming method, we present and derive the explicit forms of the robust equilibrium strategies of the firm and participants. Numerical results are also shown for the two players. We {see} that the firm and participants act very differently  in the bull and bear markets. {In addition}, ambiguity aversion influences both the firm and the participants' decisions. We have the main contributions in this paper: First, we present the robust equilibrium games of the firm and participants under two different objectives. The robust control of maximizing the probability of reaching  an upper level is relatively different from previous work. We {see} that in order to obtain the explicit solution, the penalty term for this kind of objective is also relatively different. Second, using stochastic dynamic programming method, we present the HJBI system of the robust games and derive the explicit forms of the robust investment strategy and benefit.  For the first game, we also present detailed analysis over the robust equilibrium strategies and show the Pareto optimality under some case. Third, most previous work study the robust control about the terminal wealth, which is also different from ours. Inspired by \cite{Mataramvura2008Risk}, we define the admissible sets strictly and propose another condition to replace the uniformly integrability condition to show the verification theorem. Last, numerical analyses are illustrated to {depict} the players' economic behaviours. Ambiguity aversion will always prompt the participants to ask for fewer benefits whether the economy is booming or sluggish, but this effect is minimal when the economy is in recession. For the first game, the ambiguity aversions of the two parties have different effects on the participants' strategy.
\vskip 5pt

The reminder of this paper is organized as follows. Section 2 presents the formulation of the overfunded DB pension plan under model uncertainty. In Section 3, we introduce two different robust equilibrium games for the firm and the union. Section 4 shows the admissible sets of the robust equilibrium problems and derives the explicit forms of equilibrium strategies, worst case measures and value functions. Optimality of the strategies is  {verified} strictly. The numerical results about the robust equilibrium strategies are illustrated in Section 5. Section 6 concludes this paper. Most of the proofs are in {Appendix.}

\section{\bf DB pension fund game}
In this section, we present the financial model of the overfunded DB pension fund. The sponsoring firm invests the fund surplus in the financial market and is concerned with the investment performance of the fund. To be more attractive, the overfunded pension fund provides excess benefits to the participants as a union. We consider a financial market similar with \cite{Josa2019Equilibrium} consisting of one risk-free asset and $n$ risky assets {for simplicity.}
\vskip 5pt
Consider a filtered complete probability space $\left(\Omega,\mathcal{F},\left\{\mathcal{F}_t\right\}_{t\ge 0},\mathbb{P}\right)$. $\mathcal{F}_t$ represents the information of the financial market {up to} time $t$. Different from \cite{Josa2019Equilibrium}, we suppose {that} the firm and the union are both uncertain about the financial market and $\mathbb{P}$ is the reference probability measure of the market. In order to apply the Girsanov's Theorem in infinite time horizon (see Corollary 5.2 and Proposition 5.12 of Chapter III in \cite{Ioannis1991Brownian}), the probability space is specified: $\Omega=C\left[0,\infty\right)^n$, {$W=\{W\left(t\right): t\ge 0\}$} is the coordinate mapping process on $\Omega$, $\mathbb{P}$ is the Wiener measure on $\left(\Omega,\mathcal{F}\right)$ and $\left\{\mathcal{F}_t\right\}_{t\ge 0}$ is the filtration generated by $W$. Obviously, $W=(W_1,W_2,\cdots,W_n)^T$ is an $n$-dimensional standard Brownian motion on the space $\left(\Omega,\mathcal{F},\left\{\mathcal{F}_t\right\}_{t\ge 0},\mathbb{P}\right)$.  We do not consider transaction costs in this paper. All the processes introduced {are} supposed to be well defined and adapted to $\left\{\mathcal{F}_t\right\}_{t\ge 0}$.
\subsection{\bf Financial model}
The firm sponsoring the fund invests in the financial market and expects that the fund surplus increases. For simplicity, we only consider  equity risk and suppose that there are one risk-free asset $S^0$ and $n$ risky assets $S^1,S^2\cdots,S^n$ in the financial market. The dynamics of these $n+1$ assets are
\begin{eqnarray}
	\rd S^0\left(t\right)\!&=&rS^0\left(t\right)\rd t,\!\quad S^0\left(0\right)=1,\label{equ:s0}\\
	\rd S^i\left(t\right)\!&=&S^i\!\left(t\right)\!\left(b_i\rd t+\sum_{i=1}^{n}\!\sigma_{ij}\rd W_j\left(t\right)\right),  S^i\left(0\right)=s_i,\! i=1,\cdots,n,\!\label{equ:s}
\end{eqnarray}
where  $r>0$  is the risk-free interest rate,  $b_i>r$ is the mean rate of return of the risky asset $S^i$,  {and} $\sigma_{ij}>0$ represents  the volatility coefficient. Denote {$\sigma=\left(\sigma_{ij}\right)_{n\times n}$} as the volatility matrix of the risky assets. Let $b=(b_1,b_2,\cdots,b_n)^T$ and $\vec{1}=(1,\cdots,1)^T\in\mathbf{R}^{n}$. The Sharpe ratio vector is $\theta=\sigma^{-1}(b-r\vec{1})$. {We assume that the financial market is complete, i.e., {the matrix $\Sigma\triangleq\sigma\sigma^T$} is positive definite.}
\vskip 5pt
In the financial market, we suppose that the pension fund is overfunded. Different from the management of an underfunded pension fund, the firm is endowed with a positive surplus at initial time and invests in the financial market. In order to be more attractive, the participants in  the pension fund also expect to enjoy the benefits of this positive surplus process. As in \cite{Josa2019Equilibrium}, the union of the participants claims part of the fund surplus continuously. {We assume that the number of the participants is stable, which means that we do not consider longevity risk in our work.}
\vskip 5pt
Let {$X=\{X(t): t\geq s \}$} be the surplus process of the DB pension fund with positive value {$X\left(s\right)=x$ at initial time $s\ge0$}. The firm invests an amount of {$\pi_i=\{\pi_i(t): t\geq s \}$} in $S^i$. The rest part $\pi_0=X-\sum_{i=1}^{n}\pi_i$ is allocated in the risk-free asset $S^0$. Denote $\pi=\left(\pi_1,\cdots,\pi_n\right)^T$ as the investment strategy of the firm. The union claims {benefit $P=\{P(t): t\geq s \}$} from the pension fund continuously. The firm is concerned with the surplus process $X$ while the union cares about the benefit $P$. Their decisions $\pi$ and $P$ both affect $X$. By adopting strategy $(\pi,P)$, the surplus process  $X$ is given by
\[\rd X(t)=\sum_{i=0}^n\pi_i(t)\frac{\rd S^i(t)}{S^i(t)}-P(t)\rd t,\
\
 X(s)=x. \]
As we consider a stable pension system, the surplus process totally depends on the investment performance of the firm, the benefits of the union and the initial endowment. There is a trade-off here for the benefits: {on the one hand, larger $P$ means more current benefit}, on the other hand, the fund surplus becomes small and the fund may become underfunded when $P$ is large. Substituting Eqs.~(\ref{equ:s0}) and (\ref{equ:s}) into the last equation, the surplus process {$ X$} follows
\begin{equation}
	\begin{cases}
		\rd X\left(t\right)=\left[rX\left(t\right)+\pi^T\left(t\right)
		\left(b-r\vec{1}\right)-P\left(t\right)\right]\rd t+\pi^T\left(t\right)\sigma\, \rd W\left(t\right),\\
		X\left(s\right)=x.
	\end{cases}
	\label{yuanfangcheng00}
\end{equation}
\vskip 5pt
Both the firm and the union can observe the fund surplus and know exactly the current value of the fund when making decisions. {As such}, the trading strategy $\pi$ and the benefit $P$ {depend on} the  surplus process $X$, i.e., $\pi\left(t\right)$ and $P\left(t\right)$ have the form $\pi\left(t\right)=\pi\left(t,X\left(t\right)\right)$, $P\left(t\right)=P\left(t,X\left(t\right)\right)$ at time $t$. In fact, as the system is Markovian and stationary, $\pi$ and $P$ have the form $\pi=\pi(x)$, $P=P(x)$, which is shown in  \cite{Josa2019Equilibrium}  and  means that the strategies only depend on the current wealth of the fund surplus. In our paper with ambiguity aversion, this dependence is also required, {and we will} see that the robust equilibrium strategies are of these forms.

\subsection{\bf Ambiguity aversion}
The financial model we consider is similar with the model presented  in   \cite{Josa2019Equilibrium}. However, in their work, the firm and the union can observe the financial model accurately, i.e., the probability measure $\mathbb{P}$ is the real probability measure of the financial market. However, as shown in \cite{blanchard1993movements}, the first moments of the assets are hard to be estimated accurately. When the investor makes decisions ignoring model uncertainty, he will be faced with large utility loss, see \cite{2013Robust}, \cite{huang2017robust}, etc. Ambiguity aversion is a way to characterize the decision maker's attitude towards model uncertainty. Particularly in the long term management of a DB pension fund, it is necessary to consider model uncertainty and characterize the firm (union)'s ambiguity aversion. {In what follows,} similar with \cite{Maenhout2004Robust} and \cite{Maenhout2006Robust}, we characterize model uncertainty by a set of equivalent probability measures.
\vskip 5pt
We suppose that the firm and the union both are ambiguity averse, i.e., both of them are concerned about the accuracy of the reference model. We define a set of equivalent probability measures to $\mathbb{P}$ by
\[\mathcal{Q}: =\{\mathbb{Q}|\mathbb{Q}\sim \mathbb{P}\}.\]
The firm and the union are not fully confident about the reference probability measure $\mathbb{P}$ and consider the set $\mathcal{Q}$ when making decisions. Each probability measure $\mathbb{Q}\in \mathcal{Q}$ represents a possible actual financial model.
\vskip 5pt
For $s\ge0$ and $x>0$, let $\E_{s,x}Y$ denote the expectation of random variable $Y$ under the condition $X\left(s\right)=x$. {Especially}, when $s=0$, it is written as $\E_xY$ for short. In order to show $\mathcal{Q}$ more clearly, we consider a set of $n$-dimensional-valued {processes $\mathcal{H}\triangleq \{h=\{h(t): t\geq s\}\}$} satisfying:
\begin{enumerate}
	\item {$h$} is progressively measurable w.r.t. the filtration  $\{\mathcal{F}_t\}_{t\geq s}$.
	\item For $x>0$ with initial condition $X(s)=x$, the {following} Novikov's condition holds:
	\[
	\E_{s,x}\left[\exp\left\{\int_{s}^{T}\frac{h^T\left(t\right)h\left(t\right)}{2}\rd t\right\}\right]<\infty, \quad \forall T\ge s.
	\]
\end{enumerate}
{Based on} Corollary 5.2 and Proposition 5.12 in Chapter III of  \cite{Ioannis1991Brownian}, {for $h\in \mathcal{H}$}, there is a unique equivalent  measure $\mathbb{Q}^h$ of $\mathbb{P}$ on the space $\left(\Omega,\mathcal{F},\left\{\mathcal{F}_t\right\}_{t\ge s}\right)$ with
\[\frac{\mathrm{d}\mathbb{Q}^h}{\mathrm{d}\mathbb{P}}|_{\mathcal{F}_t}=\Theta^h(t),\]
where {$\{\Theta^h(t): t\geq s\}$ is a}  $\mathbb{P}-$martingale {defined} by
\begin{equation*}
	\begin{split}
		\Theta^h(t)=\exp\left\{\int_s^th(s)^T\mathrm{d}W(s)
		-\frac{1}{2}\int_s^th(s)^Th(s)\mathrm{d}s\right\}.
	\end{split}
\end{equation*}
 {Based on} Girsanov's Theorem, the following process is {an} $n$-dimensional standard Brownian motion under $\mathbb{Q}^h$
\[\rd W^h\left(t\right)=\rd W\left(t\right)+h\left(t\right)\rd t.\]
\begin{remark}
	In the following, the expectation under the probability measure $\mathbb{Q}^h$ is denoted by $\E^h_{s,x}[\cdot]$ (or $\E^h_x[\cdot]$).
\end{remark}
{Under the  $\mathbb{Q}^h$, the fund surplus $X$ follows
\begin{eqnarray}\label{xinfangcheng0}
	\begin{cases}
		\rd X\left(t\right)	=&\left[rX\left(t\right)+\pi^T\left(t\right)
\left(b-r\vec{1}\right)-P\left(t\right)-\pi^T\left(t\right)\sigma h\left(t\right)\right]\rd t\\
&+\pi^T\left(t\right)\sigma\, \rd W^h\left(t\right),\\
X\left(s\right)=&x.		
	\end{cases}
\end{eqnarray}}
\vskip 5pt
Different from the financial model in \cite{Josa2019Equilibrium}, the wealth process includes an additional process $h$ to describe the model uncertainty. Meanwhile, in our framework, we are mainly concerned with model uncertainty of the first moments of the yields of the stock and $h$ only affects the drift terms in the wealth process. There are two players in our financial model and they have different ambiguity aversions. We denote processes $h_F$ {and}  $h_U$ as the measure transformation processes of the firm and the union, respectively.
\vspace{0.5cm}
\section{\bf Robust equilibrium games}
In this section, we present two different robust games for the firm and the union. The strategies taken by the firm and the union {affect} the fund surplus. The goals of the firm and the union are relatively opposite. The union expects to claim more benefits from the fund while the firm tries to increase the wealth of the fund. {As such}, they form a non-zero-sum game when making decisions and we are concerned with their Nash equilibrium strategies. Generally for an insurance company, there are two optimization rules: maximizing the expected wealth and minimizing the probability of ruin, see \cite{1995Optimal}. In \cite{Josa2019Equilibrium}, they show two similar different objectives for the firm. The union often does not care about the stability of the fund and aims to maximize the total expected future benefits from the fund. In the following, we present two different objectives for the firm and one objective for the union.

\subsection{\bf The first game: Maximizing the expected wealth}\label{sbs:ga}
In some cases, the firm is concerned with the expected utility of the fund surplus. When ignoring ambiguity, given $P$, the firm chooses $\pi$ to maximize the payoff
\begin{equation}
	j_F\left(s,x;\pi,P\right)=\E_{s,x}\left[\int_{s}^{\infty}e^{-\beta t}v\left(X\left(t\right)\right)\rd t\right], \label{jf1}
\end{equation}
where $v\left(\cdot\right)$ is a utility function and $\beta>0$ is the time preference of the firm. Besides, given $\pi$, the union seeks $P$ to maximize the payoff
\begin{equation}
	j_U\left(s,x;\pi,P\right)=\E_{s,x}\left[\int_{s}^{\infty}e^{-\alpha t}u\left(P\left(t\right)\right)\rd t\right],
	\label{ju1}
\end{equation}
where $u\left(\cdot\right)$ is a utility function and $\alpha>0$ is the time preference of the union.

The two objectives (\ref{jf1}) and (\ref{ju1}) have been shown in \cite{Josa2019Equilibrium}. They also obtain the explicit forms of equilibrium strategies $(\pi^*,P^*)$ under these two objectives. In our work, we are interested in the effect of model uncertainty on the firm and the union's behaviours. {Next}, we present the objectives under model uncertainty. In spite that the firm and the union are ambiguity averse, the  original probability measure $\mathbb{P}$ has reference value and the new measure $\mathbb{Q}^h$ can not deviate too far from the original measure. As in \cite{Maenhout2004Robust}, \cite{Maenhout2006Robust}, \cite{Ailing2020Optimal}, etc.,  we add a penalty term after the expected discounted utility $j_U\left(s,x;\pi,P\right)$ ($j_F\left(s,x;\pi,P\right)$) as the payoff function under ambiguity, {as such, the} {payoff of the union under ambiguity is}
\begin{eqnarray}
	J_U\left(s,x,\pi,P,h_U\right):=\E_{s,x}^{h_U}\left[\int_{s}^{\infty}e^{-\alpha t}u\left(P\left(t\right)\right)\rd t \!+\!\int_{s}^{\infty}e^{-\alpha t}\frac{\frac{1}{2}h_U^T\left(t\right)h_U\left(t\right)}{\varphi_U\left(t,X\left(t\right)\right)}\rd t\right].
	\label{ju2}
\end{eqnarray}
The increasing rate of the relative entropy from time $t$ to time $t+\mathrm{d}t$ equals $h^T(t)h(t)$, which has been shown in \cite{2018RobustLi}. $\varphi_U\left(t,X\left(t\right)\right)$ represents the degree of ambiguity aversion. When $\varphi_U\left(t,X\left(t\right)\right)$ increases, the union has less confidence about the reference model.

In this paper, $u\left(P\right)$ is the CRRA utility function, i.e., $$u\left(P\right)=\frac{P^{1-\gamma}}{1-\gamma},\quad\gamma>0,\quad \gamma\neq1,$$ where $\gamma$ represents the risk aversion of the union over financial risk. Besides, we suppose
$$\varphi_U\left(t,x\right)=\frac{\lambda}{\left(1-\gamma\right)J_U\left(t,x,\pi,P,h_U\right)},$$ where $\lambda>0$ is the ambiguity aversion parameter of the union. $\varphi_U\left(t,x\right)$ is positive and is inversely proportional to the payoff function.  When $\lambda$ increases, the union has less faith in the reference model.

\begin{remark}
	It is worth noting here that $X$ actually depends on the strategy variable $(\pi,P)$, and $\varphi_U\left(t,x\right)$ also depends on the strategy variable $(\pi,P)$ and the process $h_U$. However, for simplicity, $(\pi,P)$ is omitted and not shown in $X$, so as the following $\hat{\varphi}_U\left(x\right),\varphi_F\left(t,x \right)$ and $\hat{\varphi}_F\left(x\right)$.
\end{remark}
\vskip 5pt
The union first searches the worst case among all possible  equivalent probability measures, i.e., chooses $h^U$ to minimize the payoff function. {Then the union maximizes the payoff function under worst case scenario, which is called the objective function.} The goal of the union is to search the optimal benefit  under worst case scenario as follows
\begin{equation}
	V_U^{\pi}\left(s,x\right)\triangleq\sup_{P:\left(\pi,P\right)\in\Lambda}\inf_{h_U\in\mathcal{H}_U\left(\pi,P\right)}J_U\left(s,x,\pi,P,h_U\right).
\end{equation}
The above $\Lambda,\mathcal{H}_U\left(\pi,P\right)$ and the following $\mathcal{H}_F\left(\pi,P\right)$ represent the admissible sets of $(\pi,P)$, $h_U$ and $h_F$, respectively. The details of the definitions are given in {Subsection} \ref{admissible}.
\begin{remark}\label{rem1}
	If we assume that $P\left(t\right)=P\left(X\left(t\right)\right)$ and $h_U\left(t\right)=h_U\left(X\left(t\right)\right)$ as in the last section, {and set $\varphi_U\left(t,x\right)=e^{\alpha t}\hat{\varphi}_U\left(x\right)$, then} we have
	\begin{eqnarray*}
		J_U\left(s,x,\pi,P,h_U\right)\triangleq\E_{s,x}^{h_U}\left[\!\int_{s}^{\infty}e^{-\alpha t}\frac{P\left(X\left(t\right)\right)^{1-\gamma}}{1-\gamma}\rd t\!+\!\int_{s}^{\infty}e^{-\alpha t}\frac{\frac{1}{2}h_U^T\left(X\left(t\right)\right)h_U\left(X\left(t\right)\right)}{\hat{\varphi}_U\left(X\left(t\right)\right)}\rd t\!\right],
	\end{eqnarray*}
	which {implies}
	\begin{equation}
		J_U\left(s,x,\pi,P,h_U\right)=e^{-\alpha s}J_U\left(0,x,\pi,P,h_U\right).
		\label{shi1}
	\end{equation}
	Further {assuming} $\hat{\varphi}_U\left(x\right)=\frac{\lambda}{\left(1-\gamma\right)J_U\left(0,x,\pi,P,h_U\right)}$
	, {we have} $\varphi_U\left(t,x\right)=\frac{\lambda}{\left(1-\gamma\right)J_U\left(t,x,\pi,P,h_U\right)}$.
	Eq.~(\ref{shi1}) shows that the payoff function may be {``time-homogeneous"}, which helps us guess {the form of the solution for}  the HJBI equations in the next section.
	\label{rem2}
\end{remark}

\vspace{0.5cm}

The firm is also uncertain about the financial market and searches the optimal investment strategies under worst case scenario.  Similarly, the robust stochastic control problem for the firm is formulated as {follows:}
\begin{equation*}
	V_F^P\left(s,x\right)\triangleq\sup_{\pi:\left(\pi,P\right)\in\Lambda}
	\inf_{h_F\in\mathcal{H}_F\left(\pi,P\right)}J_F\left(s,x,\pi,P,h_F\right),
\end{equation*}
where
\begin{eqnarray}\label{jf22}
	J_F\left(s,x,\pi,P,h_F\right)&\triangleq&\E_{s,x}^{h_F}\left[\int_{s}^{\infty}e^{-\beta t}v\left(X\left(t\right)\right)\rd t+\int_{s}^{\infty}e^{-\beta t}\frac{\frac{1}{2}h_F^T\left(t\right)h_F\left(t\right)}{\varphi_F\left(t,X\left(t\right)\right)}\rd t\right].\nonumber
	\\
\end{eqnarray}
We {assume} that the firm also has a CRRA utility preference, i.e.,  $$v\left(x\right)=\frac{x^{1-\delta}}{1-\delta},\quad \delta>0,\quad \delta\neq 1,$$
where $\delta$ characterizes the firm's risk aversion of financial risk. The penalty term is given by
$$\varphi_F\left(t,x \right)=\frac{\mu}{\left(1-\delta\right)J_F \left(t,x,\pi,P,h_F\right)},$$ where $\mu>0$ is the ambiguity aversion parameter of the firm. When $\mu$ increases, the firm is more ambiguity averse towards model uncertainty. We {see} that  $\varphi_F\left(t,x\right)$ is also positive and inversely proportional to the payoff function $J_F$.
\begin{remark}
	Under the same assumption as in Remark \ref{rem1}, we {have}
	\begin{equation}
		J_F\left(s,x,\pi,P,h_F\right)=e^{-\beta s}J_F\left(0,x,\pi,P,h_F\right).
		\label{shi2}
	\end{equation}
	\label{rem3}
\end{remark}
There are two players (firm, union) in our financial model. Each of them has a strategy that affects the same fund surplus. Their goals are totally different: the firm cares about the wealth of the fund while the union cares about the benefits from the fund. They form a non-zero-sum game  and we summarize the optimization goals of them as the following robust stochastic control problem (PI)
\begin{eqnarray*}
	V_U^{\pi}\left(s,x\right)&\triangleq&\sup_{P:\left(\pi,P\right)\in\Lambda}\inf_{h_U\in\mathcal{H}_U\left(\pi,P\right)}J_U\left(s,x,\pi,P,h_U\right),
	\\
	V_F^P\left(s,x\right)&\triangleq&\sup_{\pi:\left(\pi,P\right)\in\Lambda}\inf_{h_F\in\mathcal{H}_F\left(\pi,P\right)}J_F\left(s,x,\pi,P,h_F\right).
\end{eqnarray*}
In Problem (PI), the firm's strategy $\pi$ interacts with the union's strategy $P$ to jointly affect the max-min payoff functions of both players: $V_U^{\pi}\left(s,x\right)$ and $V_F^P\left(s,x\right)$. 
{Here,} we present the definitions of the {robust equilibrium strategies} and value functions of Problem (PI).

\begin{Def}
	We call $\left(\pi_*,P_*\right)\in\Lambda$ the robust equilibrium strategy and $h_U^*\in\mathcal{H}_U\left(\pi_*,P_*\right)$, $h_F^*\in\mathcal{H}_F\left(\pi_*,P_*\right)$ the worst case measure transformation processes, if they {satisfy}
	\begin{eqnarray*}
		V_U^{\pi_*}\left(s,x\right)&=&\inf_{h_U\in \mathcal{H}_U\left(\pi_*,P_*\right)}J_U\left(s,x,\pi_*,P_*,h_U\right)=J_U\left(s,x,\pi_*,P_*,h_U^*\right),\\
		V_F^{P_*}\left(s,x\right)&=&\inf_{h_F\in \mathcal{H}_F\left(\pi_*,P_*\right)}J_F\left(s,x,\pi_*,P_*,h_F\right)=J_F\left(s,x,\pi_*,P_*,h_F^*\right).
	\end{eqnarray*}
	The value functions are given by $V_U\left(s,x\right)=V_U^{\pi_*}\left(s,x\right),V_F\left(s,x\right)=V_F^{P_*}\left(s,x\right)$,  for $\left(s,x\right)\in\mathbf{R}_+^2$.
	\label{definition of value function}
\end{Def}
In the first game,  {under worst case scenario, the firm aims to maximize the expected utility of the fund surplus while the union expects to maximize the expected utility of benefit.}  Different from  \cite{Josa2019Equilibrium}, ambiguity aversion is considered in this non-zero-sum game.
\vspace{0.5cm}
\subsection{\bf The second game: Maximizing the probability}\label{sec:game2}
As stated in the beginning of this section, some firms are also concerned with the probability of ruin. The firm may be more interested to keep the fund sustainable and overfunded.
\vskip 5pt
In \cite{Josa2019Equilibrium}, they think that the firm may also be concerned with preventing the fund surplus falling below a low level. We suppose that the firm has two fund surplus levels, $l>0$ (the lower level) and $v>0$ (the upper level). The  initial wealth of the fund is $X\left(s\right)=x\in\left(l,v\right), s\ge 0$ and the firm expects to maximize the probability of reaching $v$ before  {hitting} $l$. In the second game, the union still cares about the expected utility of the benefits.
\vskip 5pt

Next, we present the payoff functions of the firm and the union without ambiguity first. Then we introduce ambiguity in our model and show the robust control problem in the second game. Given $P$, the  firm chooses $\pi$ to maximize {the  following probability of the fund surplus reaching $v$ before  hitting $l$}:
\begin{equation*}
	j_{\bar{F}}\left(s,x;\pi,P\right)=\mathbb{P}(\mathcal{T}_v<\mathcal{T}_l|
	X(s)=x), l<x<v,
\end{equation*}
where $\mathcal{T}_z$ represents the first time that $X$ hits the value $z\geq0$, i.e., $\mathcal{T}_z=\inf
\{t>s: X(t)=z\}$. Then the {last equation is rewritten as follows:}

\begin{equation}\nonumber
	j_{\bar{F}}\left(s,x;\pi,P\right)=\E_{s,x}h\left(X\left(\mathcal{T}\right)\right),
\end{equation}
with $\mathcal{T}=\mathcal{T}_l\wedge \mathcal{T}_v$ (we need to verify $\mathcal{T}<\infty, \mathbb{P}\text{-a.s.}$ in advance) and $h\left(l\right)=0$, $h\left(v\right)=1$.

Given $\pi$, the union seeks $P$ to maximize the payoff
\begin{equation}\label{equ:ju2g}
	{j_{\bar{U}}\left(s,x;\pi,P\right)}=\E_{s,x}\left[\int_{s}^{\mathcal{T}}e^{-\alpha t}u\left(P\left(t\right)\right)\rd t+e^{-\alpha \mathcal{T}}g\left(X\left(\mathcal{T}\right)\right)\right],
\end{equation}
where $u\left(\cdot\right)$ is a utility function, $g\left(\cdot\right)$ is a bequest function, and $\alpha>0$ is the time preference of the union. In \cite{Josa2019Equilibrium}, they show that under the preference of Eq.~(\ref{equ:ju2g}), explicit solutions are unable to be derived even for very simple bequest functions $g$. Therefore, they approximate the union's payoff with the expected utility on the interval $\left[s,\infty\right)$ for a large value of $v$ as
\begin{equation}\nonumber
	{j_{\bar{U}}\left(s,x;\pi,P\right)}=\E_{s,x}\left[\int_{s}^{\infty}e^{-\alpha t}u\left(P\left(t\right)\right)\rd t\right],
\end{equation}
which is exactly the same as the payoff described by Eq.~(\ref{ju1}) in the last subsection \ref{sbs:ga}.

{Next,} we show the objectives of the firm and the union under ambiguity aversion. The objective $J_{\bar{U}}\left(s,x,\pi,P,h_{\bar{U}}\right)$ of the union under ambiguity  is the same as Eq.~(\ref{ju2}), i.e.,

\begin{eqnarray*}
	J_{\bar{U}}\left(s,x,\pi,P,h_{\bar{U}}\right)\triangleq\E_{s,x}^{h_{\bar{U}}}\left[\int_{s}^{\infty}e^{-\alpha t}u\left(P\left(t\right)\right)\rd t \!+\!\int_{s}^{\infty}e^{-\alpha t}\frac{\frac{1}{2}h_{\bar{U}}^T\left(t\right)h_{\bar{U}}\left(t\right)}{\varphi_{\bar{U}}\left(t,X\left(t\right)\right)}\rd t\right],
\end{eqnarray*}
where  $u(\cdot)$ and $\varphi_{\bar{U}}(\cdot,\cdot)$ are all the same as in the last subsection.

\vskip 5pt
The firm aims to maximize the probability under worst case scenario. {As such,} the {preference} of the firm under ambiguity aversion is as follows:
\begin{eqnarray}
J_{\bar{F}}\left(s,x,\pi,P,h_{\bar{F}}\right){\triangleq}\E_{s,x}\left[h\left(X\left(\mathcal{T}\right)\right)+\int_{s}^{\mathcal{T}}\frac{\frac{1}{2}h_{\bar{F}}^T\left(t\right)h_{\bar{F}}\left(t\right)}{\varphi_{\bar{F}}\left(t,X\left(t\right)\right)}\rd t\right],\label{jf222}
\end{eqnarray}
with $$\varphi_{\bar{F}}\left(t,x \right)=\frac{\mu}{J_{\bar{F}} \left(t,x,\pi,P,h_{\bar{F}}\right)+c},$$ where $c$ is an undetermined constant and $\mu>0$ is the ambiguity aversion parameter of the firm. The larger $\mu$ is, the  higher the ambiguity aversion level of the firm is. Then $\varphi_{\bar{F}}\left(t,x\right)$ is positive and inversely proportional to the payoff function plus the constant $c$.
\begin{remark}
	Initially, we set $c=0$. However, the homogeneity of the HJBI equations does not hold and explicit solutions can not be obtained. In order to keep the homogeneity of the system, some constant $c$ needs to be added in the penalty term.	
\end{remark}
If we assume  $h_{\bar{F}}\left(t\right)=h_{\bar{F}}\left(X\left(t\right)\right)$ {and} set $\varphi_{\bar{F}}\left(t,x\right)=\hat{\varphi}_{\bar{F}}\left(x\right)$, then we have
\begin{equation}
	J_{\bar{F}}\left(s,x,\pi,P,h_{\bar{F}}\right)=J_{\bar{F}}\left(0,x,\pi,P,h_{\bar{F}}\right).
	\label{shi3}
\end{equation}
Further {assuming}  $\hat{\varphi}_{\bar{F}}\left(x\right)=\frac{\lambda}{J_{\bar{F}}\left(0,x,\pi,P,h_{\bar{F}}\right)+c}$
, {we have} $\varphi_{\bar{F}}\left(t,x\right)=\frac{\lambda}{J_{\bar{F}}\left(t,x,\pi,P,h_{\bar{F}}\right)+c}$. {As such}, Eq.~(\ref{shi3})  {shows} that the payoff function has ``time-homogeneity", which is useful to derive the optimal solution.
\vskip 5pt
In the second game, the firm is concerned with the probability hitting some level while the union still hopes to attain more benefits. On the one hand, with larger benefits, the union has larger utility. On the other hand, the probability reaching the upper level decreases with larger benefits. {As such}, the firm and the union compete and both  affect the fund surplus. The goal of the firm and the union is to search the robust Nash equilibrium strategy $(\pi,P)$. The robust stochastic problem (PI1) of the firm and the union is as follows:
\begin{eqnarray*}
	V_{\bar{U}}^{\pi}\left(s,x\right)&\triangleq&\sup_{P:\left(\pi,P\right)\in\Lambda}\inf_{h_{\bar{U}}\in\mathcal{H}_{\bar{U}}\left(\pi,P\right)}J_{\bar{U}}\left(s,x,\pi,P,h_{\bar{U}}\right),
	\\
	V_{\bar{F}}^P\left(s,x\right)&\triangleq&\sup_{\pi:\left(\pi,P\right)\in\Lambda}\inf_{h_{\bar{F}}\in\mathcal{H}_{\bar{F}}\left(\pi,P\right)}J_{\bar{F}}\left(s,x,\pi,P,h_{\bar{F}}\right),
\end{eqnarray*}
where the admissible sets of  $\Lambda,\mathcal{H}_{\bar{U}}\left(\pi,P\right)$ and $\mathcal{H}_{\bar{F}}\left(\pi,P\right)$ are all given in the next section.
\vskip 5pt
The definitions of the robust equilibrium strategy, the worst case measure transformation process, and value functions of Problem (PI1) are exactly the same as in Definition \ref{definition of value function}.
\vskip 10pt
\section{\bf Main results}
In this section, we derive the robust equilibrium strategies of Problems (PI) and (PI1) separately. {Using} stochastic dynamic programming method, we present the HJBI equations for these two games. Explicit forms of the robust equilibrium strategies are obtained by the HJBI equations. However, the optimality of the robust strategy is not easy to be proved. In order to show the optimality, first we present the definitions of  admissible sets to ensure the well-posedness and the integrability {(or finiteness)} of the problem. Then we show the optimality of the explicit solutions in detail.
\vskip 5pt
\subsection{\bf The first game}
Problem (PI) is concerned with the utilities of the wealth and benefit of the fund  surplus. We first define the admissible sets for $(\pi,P)$, $h_U$ and $h_F$. Then we search the robust equilibrium strategy and the worst case measure transformation process within the admissible sets. Detailed proofs of the optimality of the robust equilibrium strategy are also presented. {In addition}, we show that  the robust equilibrium strategy is also Pareto optimal {in some cases}.
\subsubsection{\bf Admissible sets}\label{admissible}
For  fixed $s\ge 0$ and $x>0$, in order to ensure the conditions in the Girsanov's Theorem, integrability of the objective function and the validity of the verification theorem, we need to introduce the admissible sets of  the strategy $(\pi,P)$ and corresponding measure transformation process $(h_F, h_U)$. First we define the admissible set of the strategy $(\pi,P)$.
\begin{Def}
	A strategy {$(\pi,P)=\left\{\left(\pi\left(t\right),P\left(t\right )\right ): t\ge s\right\}$} is admissible, if it satisfies
	\begin{enumerate}	
		\item[(i)]{$\left(\pi,P\right)$} is  adapted to the filtration $\{\mathcal{F}_t\}_{t\ge s}$ and $\pi\left(t\right)>0,P\left(t\right)> 0,\forall t\ge s, \mathbb{P}\text{-a.s.}$.
		\item[(ii)]Under the strategy $(\pi,P)$,  SDE~(\ref{yuanfangcheng00}) has a pathwise unique solution $X$  with $X\left(t\right)>0,\forall t\ge s, \mathbb{P}\text{-a.s.}$.
	\end{enumerate}
	The set of  admissible strategies $\left(\pi,P\right)$ is denoted {by} $\Lambda$.
	\label{def2.1}
\end{Def}
{Next,} we {introduce} the admissible set of the measure transformation processes $h_U$ and $h_F$.
\begin{Def}
	For every fixed admissible strategy $\left(\pi,P\right)$, $X$ is the unique solution of SDE~(\ref{yuanfangcheng00}) under the strategy $\left(\pi,P\right)$.
	We say that {$h_U=\left\{h_U\left(t\right): t\geq s\right\}$} is admissible  about $\left(\pi,P\right)$, if $X\left(t\right)>0$, $\pi\left(t\right)>0$, $P\left(t\right)>0$, $\forall t\ge s$, $\mathbb{Q}^{h_U}\text{-a.s.}$ and $h_U$ satisfies
	\begin{enumerate}
		\item[(i)]
		\begin{equation}\nonumber
			\E_{s,x}\left[\exp\left\{\int_{s}^{T}\frac{h_U\left(t\right)h_U\left(t\right)}{2}\rd t\right\}\right]<\infty, \quad \forall T\ge s,
		\end{equation}
		\item[(ii)]
		\begin{equation}\nonumber
			\E_{s,x}^{h_U}\left[\int_{s}^{\infty}\left| e^{-\alpha t}\frac{\left[P\left(t\right)\right]^{1-\gamma}}{1-\gamma}\right| \rd t\right]<\infty,
		\end{equation}
		\item[(iii)]
		\begin{equation}\nonumber
			\lim_{t\rightarrow +\infty}\E_{s,x}^{h_U}\left[e^{-\alpha t}\frac{\left[X\left(t\right)\right]^{1-\gamma}}{1-\gamma}\right]=0,
		\end{equation}
		\item[(iv)]
		\begin{equation}\nonumber
			\E_{s,x}^{h_U}\left[\int_{s}^{T}\left(X\left(t\right)\right)^{-2\gamma}\pi^T\left(t\right)\pi\left(t\right)\rd t\right]<\infty, \quad \forall T\ge s.
		\end{equation}
	\end{enumerate}
	Denote $\mathcal{H}_U\left(\pi,P\right)$ as the set of all admissible processes $h_U$ about $\left(\pi,P\right)$.\label{def2.12}
\end{Def}	
Similarly, we {give} the set of the admissible measure transformation process $h_F$ for the firm as follows.
\begin{Def}
	For every fixed admissible strategy $\left(\pi,P\right)$, $X$ is the unique solution of SDE~(\ref{yuanfangcheng00}) under the strategy $\left(\pi,P\right)$. We say that {$h_F=\left\{h_F\left(t\right): t\geq s\right\}$} is admissible about $\left(\pi,P\right)$, if $X\left(t\right)>0, \pi\left(t\right)>0, P\left(t\right)>0, \forall t\ge s, \mathbb{Q}^{h_F}\text{-a.s.}$ and  $h_F$ satisfies
	\begin{enumerate}
		\item[(i)]
		\begin{equation}\nonumber
			\E_{s,x}\left[\exp\left\{\int_{s}^{T}\frac{h_F\left(t\right)h_F\left(t\right)}{2}\rd t\right\}\right]<\infty, \quad \forall T\ge s,
		\end{equation}
		\item[(ii)]
		\begin{equation}\nonumber
			\E_{s,x}^{h_F}\left[\int_{s}^{\infty}\left|e^{-\beta t}\frac{\left[X\left(t\right)\right]^{1-\delta}}{1-\delta}\right|\rd t\right]<\infty,
		\end{equation}
		\item[(iii)]
		\begin{equation}\nonumber
			\lim_{t\rightarrow +\infty}\E_{s,x}^{h_F}\left[e^{-\beta t}\frac{\left[X\left(t\right)\right]^{1-\delta}}{1-\delta}\right]=0,
		\end{equation}
		\item[(iv)]
		\begin{equation}\nonumber
			\E_{s,x}^{h_F}\left[\int_{s}^{T}\left(X\left(t\right)\right)^{-2\delta}\pi^T\left(t\right)\pi\left(t\right)\rd t\right]<\infty, \quad \forall T\ge s.
		\end{equation}
	\end{enumerate}
	All admissible processes $h_F$ about $\left(\pi,P\right)$ form an admissible set denoted by  $\mathcal{H}_F\left(\pi,P\right)$.
	\label{def2.2}
\end{Def}
Condition (i) in Definitions \ref{def2.12} and \ref{def2.2} is the Novikov's condition, which guarantees the validity of the measure transformation from reference measure $\mathbb{P}$ to $\mathbb{Q}^h$. The integrability of the penalty terms in Eqs.~(\ref{ju2}) and (\ref{jf22}) is ensured in Condition (ii) of  Definitions \ref{def2.12} and \ref{def2.2}. Conditions (iii) and (iv) are proposed to ensure the optimality of the robust equilibrium {strategy.}

\vskip 5pt
\begin{remark}In most stochastic dynamic programming problems, if the problem is formulated in infinite time, Condition (iii) is always required to ensure {the optimality}. If the problem is formulated in finite time horizon, we do not need to introduce Condition (iii) of Definitions  \ref{def2.12} and \ref{def2.2}, see \cite{Mataramvura2008Risk} for example.
Besides, we introduce the integrability condition (iv) of   Definitions  \ref{def2.12} and \ref{def2.2} in our work. {In \cite{korn2002stochastic} and \cite{kraft2012optimal},} the uniformly integrability condition is required and localization method is applied to show the optimality. In our work, the robust control problem has two features: {On the one hand, the sign of the selected utility function does not change and then the monotonic convergence theorem can be established; on the other hand, this problem is formulated in infinite time horizon and does not involve the stopping time.} Therefore, in our model, the integral expansion theorem can be established and we can replace the uniformly integrability condition by Condition (iv) to avoid the use of localization method.
	
	\label{remh2}
\end{remark}

In fact, the admissible set $\Lambda$ and admissible measure transformation sets $\mathcal{H}_U\left(\pi,P\right)$ and $\mathcal{H}_F\left(\pi,P\right)$ about $\left(\pi,P\right)\in\Lambda$  have already been shown in Problem (PI) of Subsection \ref{sbs:ga}. In this subsection, the detailed descriptions of the sets are presented. Some conditions about the sets ensure the well-posedness of the robust control problem. The other conditions are useful to show the optimality of the derived strategy.
 Now we derive the robust equilibrium strategy under worst case scenario and the value functions of  Problem (PI). {In addition}, the optimality is also shown precisely.
\vskip 5pt
\subsubsection{\bf Robust equilibrium strategy}
{In this subsection}, we solve the robust stochastic control problem (PI). {First}, we derive the associated HJBI equations to Problem (PI) based on stochastic dynamic programming method. {Second}, we guess the forms of the value functions and obtain the explicit forms of the robust equilibrium strategy as well as the  processes associated with the  worst case measures. Last, we prove the optimality of the equilibrium strategy in detail. The derivation and solving procedure of the HJBI equations is similar with many previous work. The main contribution here is to verify the optimality strictly.
\vskip 5pt
$\forall f(\cdot,\cdot)\in C^{1,2}\left(\mathbf{R}_+^2\right)$, the infinitesimal generator $\mathcal{A}^{\pi,P,h}$ {of the surplus process $X$  defined by (\ref{xinfangcheng0}) is}
\begin{eqnarray}
	\mathcal{A}^{{\pi},{P},{h}}f\left(s,x\right)&\triangleq &f_s\left(s,x\right)+\left(rx+{\pi}^T\left(b-r\vec{1}\right)-{P}-{\pi}^T\sigma h\right)f_x\left(s,x\right)\nonumber
	\\
	&&+\frac{1}{2}{\pi}^T\Sigma{\pi} f_{xx}\left(s,x\right).
	\label{suanzi}
\end{eqnarray}
The following theorem presents {the associated HJBI equations to Problem (PI)}.
\begin{proposition}[HJBI equations]Suppose that there are two functions {$W^U\left(\cdot,\cdot\right)\in C^{1,2}\left(\mathbf{R}_+^2\right)$, and} $W^F\left(\cdot,\cdot\right)\in C^{1,2}\left(\mathbf{R}_+^2\right)$. {Then the HJBI systems to Problem $\left(PI\right)$} are
	\begin{equation}
		0=\sup_{\hat{P}}\inf_{\hat{h}_U}\left
\{\mathcal{A}^{\hat{\pi},\hat{P},\hat{h}_U}
W^U\left(s,x\right)+\Phi^U\left(s,x,\hat{\pi},\hat{P},\hat{h}_U,W^U\right)\right\},
		\label{hjbfc1}
	\end{equation}
	\begin{equation}
		0=\sup_{\hat{\pi}}\inf_{\hat{h}_F}
\left\{\mathcal{A}^{\hat{\pi},\hat{P},
\hat{h}_F}W^F\left(s,x\right)+\Phi^F\left(s,x,\hat{\pi},\hat{P},\hat{h}_F,W^F\right)\right\},
		\label{hjbcf2}
	\end{equation}
	where
	\begin{equation*}
		\Phi^U\left(s,x,
\hat{\pi},\hat{P},\hat{h}_U,W^U\right)\triangleq e^{-\alpha s}\frac{\hat{P}^{1-\gamma}}
{1-\gamma}+\left(1-\gamma\right)\frac{\hat{h}_U^T\hat{h}_U}{2\lambda}W^U\left(s,x\right),
	\end{equation*}
	\begin{equation*}
		\Phi^F\left(s,x,\hat{\pi},\hat{P},\hat{h}_F,W^F\right)\triangleq e^{-\beta s}\frac{x^{1-\delta}}{1-\delta}+\left(1-\delta\right)\frac{\hat{h}_F^T\hat{h}_F}
{2\mu}W^F\left(s,x\right).
	\end{equation*}
	\label{pro3.1}
\end{proposition}
\begin{proof}
The derivation is simple and we omit it here, see \cite{Maenhout2006Robust} for details.
\end{proof}
{The following theorem gives the closed-forms of solutions to the HJBI equations  in  Proposition \ref{pro3.1}.}
\begin{theorem}{Assume that the parameters $A$ and $B$ defined by Eqs.~(\ref{A}) and (\ref{B}) are positive. Then the associated HJBI equations~(\ref{hjbfc1} )-(\ref{hjbcf2}) to Problem (PI)  have  solutions in $C^{1,2}\left(\mathbf{R}_+^2\right)$ :}
	\begin{eqnarray}
		W^U\left(s,x\right)&=&Ae^{-\alpha s}\frac{x^{1-\gamma}}{1-\gamma},
		\label{jieu}
		\\
		W^F\left(s,x\right)&=&Be^{-\beta s}\frac{x^{1-\delta}}{1-\delta}.
		\label{jief}
	\end{eqnarray}
	Moreover, for fixed $\left(s,x\right)$, the maximum point $\hat{P}_*$ and the minimum point $\hat{h}_{U*}$ in Eq.~(\ref{hjbfc1}) are given by
	\begin{eqnarray}
		\hat{P}_*&=&A^{-\frac{1}{\gamma}}x,
		\label{p}
		\\
		\hat{h}_{U*}&=&\frac{\lambda}{\mu+\delta}\theta,
		\label{hu}	
	\end{eqnarray}
	and the maximum point $\hat{\pi}_*$ and the minimum point $\hat{h}_{F*}$ in Eq.~(\ref{hjbcf2}) are  given by
	\begin{eqnarray}
		\hat{\pi}_*&=&\frac{1}{\mu+\delta}\Sigma^{-1}\left(b-r\vec{1}\right)x,
		\label{pi}
		\\
		\hat{h}_{F*}&=&\frac{\mu}{\mu+\delta}\theta,
		\label{hf}
	\end{eqnarray}
{where
\begin{eqnarray}
		A&=&\left[ \frac{\alpha}{\gamma}-\frac{1-\gamma}{\gamma}\left(r+\left(\frac{1}{\mu+\delta}-\frac{\lambda+\gamma}{2\left(\mu+\delta\right)^2}\right)\theta^T\theta\right) \right]^{-\gamma},
		\label{A}
		\\
		B\!&=&\!\frac{1}{1\!-\!\delta}\left[ \frac{\beta}{1\!-\!\delta}\!+\!\frac{\alpha\!-\!r}{\gamma}\!+\!\left( \frac{\left(1\!-\!\gamma\right)\left(\lambda\!+\!\gamma\right)}
{2\gamma\left(\mu\!+\!\delta\right)^2}\!-\!\frac{1\!-\!\gamma}
{\gamma\left(\mu\!+\!\delta\right)}\!-\!\frac{1}{2\left(\mu\!+\!\delta\right)} \right)\theta^T\theta \right]^{-1}.
		\label{B}
	\end{eqnarray}}
	\label{the3.1}
\end{theorem}

\begin{proof}
See Appendix \ref{a:the3.1}.
\end{proof}
The results in Theorem \ref{the3.1} show that $\hat{h}_{U*}$ and $\hat{h}_{F*}$ are two constants while $\hat{P}_*$ and $\hat{\pi}_*$ are two  functions proportional to $x$. As such, we always use $\hat{P}_*\left(x\right)$ and $\hat{\pi}_*\left(x\right)$ to {denote} $\hat{P}_*$ and $\hat{\pi}_*$.
\vskip 5pt
We have shown the solutions of the HJBI equations in Theorem \ref{the3.1}. Next we show that the solution of the HJBI equations is just the solution of Problem (PI).
We need to prove that the  $C^{1,2}\left(\mathbf{R}_+^2\right)$ solutions  (\ref{jieu}) and (\ref{jief}) in Theorem \ref{the3.1} are exactly the value functions of Problem (PI) for the union and the firm. We also discuss the relationship between the robust equilibrium strategy and the maximum points given by Eqs.~(\ref{p}) and (\ref{pi}), the relationship between their respective worst case measures and the measures corresponding to the minimum points given by Eqs.~(\ref{hu}) and (\ref{hf}).
For simplicity,  denote $\pi\left(t\right)=\hat{\pi}_*\left(X\left(t\right)\right), P\left(t\right)=\hat{P}_*\left(X\left(t\right)\right)$, then we can prove the following lemma.
\begin{lemma}
	{Under  $\left(\pi,P\right)=\left(\left\{\hat{\pi}_*\left(X\left(t\right)\right): t\geq s\right\},\{\hat{P}_*\left(X\left(t\right)\right): t\geq s\}\right)$,} SDE~(\ref{yuanfangcheng00}) has a unique solution $X^*$.
	\label{pro3.0}
\end{lemma}
\begin{proof}
	See Appendix \ref{a:pro3.0}.
\end{proof}
 {Based on Lemma \ref{pro3.0}, we  obtain a pair of strategy $\left(\pi_*,P_*\right)$:\ $\pi_*=\left\{\hat{\pi}_*\left(X^*\left(t\right)\right): t\geq s\right \}$} and {$P_*=\{\hat{P}_*\left(X^*\left(t\right)\right): t\geq s\}$ and a pair of measure transformation process $\left(h_F^*,h_U^*\right)$ with $h_F^*\left(t\right)\equiv \hat{h}_{F*}$ and $h_U^*\left(t\right)\equiv \hat{h}_{U*}, t\ge s$.} Now we show that the strategy $\left(\pi_*,P_*\right)$  is admissible and the measure transformation processes $h_F^*$ and $h_U^*$ are admissible.
\begin{lemma}
	Suppose that the $A$ and $B$ in Theorem \ref{the3.1} are positive. Then we have  $\left(\pi_*,P_*\right)\in\Lambda$, $ h_U^*\in\mathcal{H}_U\left(\pi_*,P_*\right)$ and 	 $h_F^*\in\mathcal{H}_F\left(\pi_*,P_*\right)$.
	\label{pro3.2}
\end{lemma}
\begin{proof}
	See Appendix \ref{a:pro3.2}.
\end{proof}
{Based on Lemma \ref{pro3.2} and Theorem \ref{the3.1},  we only need to show that $W^U$, $W^F$,
$\left(\pi_*,P_*\right)$ and $(h_F^*, h_U^*) $ are the solutions of Problem (PI) in the following.}
\begin{theorem}[Verification theorem]
Suppose that the $A$ and $B$ in Theorem \ref{the3.1} are positive. Then $W^U$ and $W^F$ in Theorem \ref{the3.1} are the value functions of the union and firm, respectively.  Besides, for the robust stochastic control problem $\left(PI\right)$, $\left(\pi_*,P_*\right)$ is the robust equilibrium strategy, $\left(h_F^*, h_U^*\right)$ is  the worst case measure transformation process of the firm and the union, respectively.
{In addition}, the corresponding optimal wealth process under different probability measures $\mathbb{P}$, $\mathbb{Q}^{h_U^*}$ and $\mathbb{Q}^{h_F^*}$ can be expressed {by}

{{\begin{eqnarray}\label{zuiyoucaifu}
X^*(t)=&\!\!\!\!\!\!\!\!\!\!\!\!\!\!\!\!\!\!\!\!\!\!\!\!\!\!\!
x\exp\left\{\left(r+\frac{\theta^T\theta}
{\mu+\delta}-A^{-\frac{1}{\gamma}}-
\frac{\theta^T\theta}{2\left(\mu+
\delta\right)^2}\right)\left(t-s\right)\!\!
+\!\!\frac{1}{\mu+\delta}\theta^T\left(W(t)\!\!-\!\!W(s)\right)\right\}
				\nonumber\\
\!=&x\!\exp\!\left\{\!\left(\!r\!+\!
				\frac{\theta^T\theta}{\mu\!+\!\delta}\!-\!
				A^{-\frac{1}{\gamma}}\!\!-\!\!\theta^T\!
				\theta\frac{\lambda}{\left(\mu\!+\!\delta\right)^2}
				\!-\!\frac{\theta^T\theta}{2\left(\mu\! +\!\delta\right)^2}\right)\!\left(t-s\right)\!+\!
\frac{1}{\mu+\delta}\theta^T\!\left(W^{h_U^*}(t)\!\!-\!\!W^{h_U^*}(s)\!\right)\!\right\}
				\nonumber\\
				\!=&x\!\exp\!\left\{\!\left(\!r\!+\!\frac{\theta^T\theta}
{\mu\!+\!\delta}\!-\!A^{-\frac{1}{\gamma}}\!\!-\!\!\theta^T
\theta\frac{\mu}{\left(\mu\!+\!\delta\right)^2}\!\!-\!\!
\frac{\theta^T\theta}{2\left(\mu\!+\!\delta\right)^2}\right)
\left(t\!-\!s\right)+\frac{1}{\mu\!+\!\delta}\theta^T
\left(W^{h_F^*}(t)\!\!-\!\!W^{h_F^*}\!(s)\!\right)\!\right\}.
	\end{eqnarray}}}
	\label{the3.2}
\end{theorem}
\begin{remark}$\left(\hat{\pi}_*\left(x\right),\hat{P}_*\left(x\right)\right)$ and  $\left(\hat{h}_{F*},\hat{h}_{U*}\right)$ are the optimal feedback functions and the worst case measure transformation parameters of the firm and the union, respectively.
\end{remark}
Before proving Theorem \ref{the3.2}, we present the following lemma.
\begin{lemma}
	{Suppose that the $A$ and $B$ in Theorem \ref{the3.1} are positive, and the functions $W^U\left(\cdot,\cdot\right)$ and $W^F\left(\cdot,\cdot\right)$  in Theorem \ref{the3.1} belong to  $ C^{1,2}\left(\mathbf{R}_+^2\right)$.} Then
	\begin{enumerate}
		\item[(i)]$\forall \hat{h}_U\in\mathbf{R},\quad\forall \hat{h}_F\in\mathbf{R},\quad\forall\left(t,y\right)\in\mathbf{R}_+^2,$
		\begin{eqnarray}
			&&\mathcal{A}^{\hat{\pi}_*\left(y\right),\hat{P}_*\left(y\right),\hat{h}_U}W^U\left(t,y\right)+\Phi^U\left(t,y,\hat{\pi}_*\left(y\right),\hat{P}_*\left(y\right),\hat{h}_U,W^U\right)\ge 0,\nonumber\\
			&&\mathcal{A}^{\hat{\pi}_*\left(y\right),\hat{P}_*\left(y\right),\hat{h}_F}W^F\left(t,y\right)+\Phi^F\left(t,y,\hat{\pi}_*\left(y\right),\hat{P}_*\left(y\right),\hat{h}_F,W^F\right)\ge 0.\nonumber
		\end{eqnarray}
		\item[(ii)]$\forall \hat{P}\in\mathbf{R},\quad\forall\hat{\pi}\in\mathbf{R},\quad\forall\left(t,y\right)\in\mathbf{R}_+^2,$
		\begin{eqnarray}
			&&\mathcal{A}^{\hat{\pi}_*\left(y\right),\hat{P},\hat{h}_{U*}}W^U\left(t,y\right)+\Phi^U\left(t,y,\hat{\pi}_*\left(y\right),\hat{P},\hat{h}_{U*},W^U\right)\le 0,\nonumber\\
			&&\mathcal{A}^{\hat{\pi},\hat{P}_*\left(y\right),\hat{h}_{F*}}W^F\left(t,y\right)+\Phi^F\left(t,y,\hat{\pi},\hat{P}_*\left(y\right),\hat{h}_{F*},W^F\right)\le 0.\nonumber
		\end{eqnarray}
		\item[(iii)]$\forall
		\left(t,y\right)\in\mathbf{R}_+\times\left(0,\infty\right),$
		\begin{eqnarray}
			&&\mathcal{A}^{\hat{\pi}_*\left(y\right),\hat{P}_*\left(y\right),\hat{h}_{U*}}W^U\left(t,y\right)+\Phi^U\left(t,y,\hat{\pi}_*\left(y\right),\hat{P}_*\left(y\right),\hat{h}_{U*},W^U\right)= 0,\nonumber\\
			&&\mathcal{A}^{\hat{\pi}_*\left(y\right),\hat{P}_*\left(y\right),\hat{h}_{F*}}W^F\left(t,y\right)+\Phi^F\left(t,y,\hat{\pi}_*\left(y\right),\hat{P}_*\left(y\right),\hat{h}_{F*},W^F\right)= 0,\nonumber
		\end{eqnarray}
		\item[(iv)]$\forall\left(\pi,P\right)\in\Lambda, h_U\in\mathcal{H}_U\left(\pi,P\right),h_F\in\mathcal{H}_F\left(\pi,P\right),$
		\begin{eqnarray}
			&\lim_{T\rightarrow+\infty}E_{s,x}^{h_U}W^U\left(T,X\left(T\right)\right)=\lim_{T\rightarrow+\infty}E_{s,x}^{h_F}W^F\left(T,X\left(T\right)\right)=0,\nonumber
		\end{eqnarray}
where $X$ is the unique solution of SDE~(\ref{yuanfangcheng00}).
	\end{enumerate}
	\label{pro3.3}
\end{lemma}
\begin{proof}
	See Appendix \ref{a:pro3.3}.
\end{proof}
Next, {we return to prove Theorem \ref{the3.2}   and show} the optimality of the solutions given in Theorem \ref{the3.1}.
\begin{proof}
	$\forall\left(\pi,P\right)\in\Lambda$, $h_U\in\mathcal{H}_U$, let $X$ be the unique solution of SDE (\ref{yuanfangcheng00}). {Using} It\^{o}'s formula, we have
	{\begin{equation}
			W^U\left(T\!,\!X(T)\right)\!=\! W^U\left(s\!,\!X(s)\right)\!\!+\!\!\int_{s}^{T}\mathcal{A}^{\pi,P,h_U}W^U\left(t,X(t)\right)\rd t\!\!+\!\int_{s}^{T}\!W_x^U\left(t,X(t)\right)\!\pi\left(t\right)^T\!
			\sigma\rd W^{h_U}(t).\label{wdif}
	\end{equation}}
	As $\Sigma$ is positive definite, {$\exists$ $ C\geq 0$,$\;$ s.t.$\;$} $u^T\Sigma u\le Cu^Tu$, $\forall u\in \mathbf{R}^n$. {In addition, using}  $W_x^U\left(t,X(t)\right)=e^{-\alpha t}X(t)^{-\gamma}$, we know that
	\begin{equation}\nonumber
		\E_{s,x}^{h_U}\int_{s}^{T}\left(W_x^U\left(t,X(t)\right)\right)^2\pi\left(t\right)^T\Sigma\pi\left(t\right)\rd t<\infty,\ \ \forall T>s,
	\end{equation}
	is equivalent to Condition (iv) of Definition \ref{def2.2}. {As such, \\ $\left\{\int_{s}^{T}W_x^U\left(t,X(t)\right)\pi\left(t\right)^T\sigma \rd W^{h_U}(t): T\geq s \right\}$ is a $\mathbb{Q}^{h_U}$ martingale. Taking} expectation on both sides of Eq.~(\ref{wdif}), we have 
	\begin{equation}
		\E_{s,x}^{h_U}W^U\left(T,X(T)\right)=W^U\left(s,x\right)+\E_{s,x}^{h_U}\int_{s}^{T}\mathcal{A}^{\pi\left(t\right),
			P\left(t\right),h_U\left(t\right)}W^U\left(t,X(t)\right)\rd s.
		\label{th36}
	\end{equation}
 {Using property (i) of Lemma \ref{pro3.3},  $\pi_*\left(t\right)=\hat{\pi}_*\left(X(t)\right)$ and $P_*\left(t\right)=\hat{P}_*\left(X(t)\right)$, we have that}
	$\forall \ h_U\in \mathcal{H}_U\left(\pi_*,P_*\right)$,
{\begin{equation}
\!\!\!\!	\!\!\!\!	W^U\left(s,x\right)\!\le \! \E_{s,x}^{h_U}\int_{s}^{T}\Phi^U\left(t,X(t),\pi_*\left(t\right),P_*\left(t\right),h_U\left(t\right),W^U\right)\rd t\!+\!\E_{s,x}^{h_U}W^U\left(T,X(T)\right).
		\label{wwww}
\end{equation}}
{Applying property (iv) in Lemma \ref{pro3.3} to the last term of Eq.~(\ref{wwww}),
noticing that one term of $\Phi^U$ is non-negative while the sign of the other term is constant, based on the integral expansion theorem, letting $T\rightarrow+\infty$, we obtain
	\begin{equation}
		W^U\left(s,x\right)\le J_U\left(s,x,\pi_*,P_*,h_U\right).
		\label{budengsji1}
	\end{equation}}
{Noting that Ineq.~(\ref{budengsji1}) holds for any $h_U\in \mathcal{H}_U\left(\pi_*,P_*\right)$, we have}
	\begin{equation*}
		W^U\left(s,x\right)\le \inf_{h_U\in \mathcal{H}_U\left(\pi_*,P_*\right)}J_U\left(s,x,\pi_*,P_*,h_U\right).
	\end{equation*}
	Thus,
	\begin{equation}\label{21}
		W^U\left(s,x\right)\le \sup_{P:\left(\pi_*,P\right)\in\Lambda}\inf_{h_U\in \mathcal{H}_U\left(\pi_*,P\right)}J_U\left(s,x,\pi_*,P,h_U\right)=V_U^{\pi_*}\left(s,x\right).
	\end{equation}
	On the other hand, applying property (ii) in Lemma \ref{pro3.3} to Eq.~(\ref{th36}), we have that:
	for $(\pi_*,h_U^*)$, $\forall P$ such that $\left(\pi_*,P\right)\in\Lambda$ and $ h_U^*\in\mathcal{H}_U\left(\pi_*,P\right)$,
	{	\begin{equation*}
			W^U\left(s,x\right)\ge \E_{s,x}^{h_U^*}\int_{s}^{T}\Phi^U\left(t,X(t),\pi_*\left(t\right),P\left(t\right),h_U^*\left(t\right),W^U\left(t,X(t)\right)\right)\rd s+\E_{s,x}^{h_U^*}W^U\left(T,X(T)\right).
	\end{equation*}}
	Letting $T\rightarrow+\infty$, we have
\begin{equation}
		W^U\left(s,x\right)\ge J_U\left(s,x,\pi_*,P,h_U^*\right)\ge \inf_{h_U\in \mathcal{H}_U\left(\pi_*,P\right)}J_U\left(s,x,\pi_*,P,h_U\right).
		\label{budengshi2}
\end{equation}
	Thus,
\begin{equation}\label{22}
		W^U\left(s,x\right)\ge \sup_{P:\left(\pi_*,P\right)\in\Lambda}\inf_{h_U\in \mathcal{H}_U\left(\pi_*,P\right)}J_U\left(s,x,\pi_*,P,h_U\right)=V_U^{\pi_*}\left(s,x\right).
\end{equation}
{Applying property (iii) in Lemma \ref{pro3.3} to Eq.~(\ref{th36}), using Ineqs.~(\ref{21})-(\ref{22}) and letting $T\rightarrow+\infty$, we have}
	\begin{equation*}
		V_U^{\pi_*}\left(s,x\right)=W^U\left(s,x\right)= J_U\left(s,x,\pi_*,P_*,h_U^*\right).
	\end{equation*}
	On the one hand,
	\begin{eqnarray}\label{23}
		\sup_{P:\left(\pi_*,P\right)\in\Lambda}\inf_{h_U\in \mathcal{H}_U\left(\pi_*,P\right)}J_U\left(s,x,\pi_*,P,h_U\right)\le \inf_{h_U}\sup_{\substack{P:\left(\pi_*,P\right)\in\Lambda\\h_U\in\mathcal{H}_U\left(\pi_*,P\right)}}J_U\left(s,x,\pi_*,P,h_U\right).
	\end{eqnarray}
	On the other hand, using Ineqs.~$\left(\ref{budengsji1}\right)$ and $\left(\ref{budengshi2}\right)$,
	\begin{eqnarray*}
		&&\inf_{h_U}\sup_{\substack{P:\left(\pi_*,P\right)\in\Lambda\\h_U\in\mathcal{H}_U\left(\pi_*,P\right)}}J_U\left(s,x,\pi_*,P,h_U\right)\le \sup_{\substack{P:\left(\pi_*,P\right)\in\Lambda\\h_U^*\in\mathcal{H}_U\left(\pi_*,P\right)}}J_U\left(s,x,\pi_*,P_*,h_U\right)\le W^U\left(s,x\right)\nonumber
		\\
		&\le& \inf_{h_U\in \mathcal{H}_U\left(\pi_*,P_*\right)}J_U\left(s,x,\pi_*,P_*,h_U\right) \le\sup_{P:\left(\pi_*,P\right)\in\Lambda}\inf_{h_U\in \mathcal{H}_U\left(\pi_*,P\right)}J_U\left(s,x,\pi_*,P,h_U\right).
	\end{eqnarray*}
{Thus, using Ineq.~(\ref{23}), we obtain}
	\begin{eqnarray}
		&&\inf_{h_U\in \mathcal{H}_U\left(\pi_*,P_*\right)}J_U\left(s,x,\pi_*,P_*,h_U\right)
		=\sup_{P:\left(\pi_*,P\right)\in\Lambda}\inf_{h_U\in \mathcal{H}_U\left(\pi_*,P\right)}J_U\left(s,x,\pi_*,P,h_U\right) \nonumber
		\\&=&V_U^{\pi_*}\left(s,x\right)=W^U\left(s,x\right)= J_U\left(s,x,\pi_*,P_*,h_U^*\right).\nonumber
	\end{eqnarray}
Similarly,
	\begin{eqnarray}
		&&\inf_{h_F\in \mathcal{H}_F\left(\pi_*,P_*\right)}J_F\left(s,x,\pi_*,P_*,h_F\right)=\sup_{\pi:\left(\pi,P_*\right)\in\Lambda}\inf_{h_F\in \mathcal{H}_F\left(\pi,P_*\right)}J_F\left(s,x,\pi,P_*,h_F\right) \nonumber
		\\&=&V_F^{P_*}\left(s,x\right)=W^F\left(s,x\right)= J_F\left(s,x,\pi_*,P_*,h_F^*\right).\nonumber
	\end{eqnarray}
{Therefore}, $W^U\left(\cdot,\cdot\right)=V_U\left(\cdot,\cdot\right),W^F\left(\cdot,\cdot\right)=V_F\left(\cdot,\cdot\right)$ are the value functions of the union and the firm. {In addition,}
	the strategies {$P_*$ and $\pi_*$} are the robust equilibrium strategies of  the union and the firm, respectively. Moreover, $h_U^*$ and $h_F^*$ are the worst case measure transformation processes of the union and the firm{, respectively.} The explicit form of the optimal wealth process is shown in the proof of Lemma \ref{pro3.0}.
\end{proof}
\subsubsection{\bf Static analysis of the equilibrium strategy}\label{RA}
Theorem \ref{the3.2} presents the optimal solution of Problem (PI). We are most concerned with the optimal feedback functions (in other words, the robust equilibrium strategies): the union's optimal feedback function $\hat{P}_*\left(x\right)=A^{-\frac{1}{\gamma}}x$ and the firm's optimal feedback function $\hat{\pi}_*\left(x\right)=\frac{1}{\mu+\delta}\Sigma^{-1}\left(b-r\vec{1}\right)x$. As in \cite{Josa2019Equilibrium}, the functions are all proportional functions of the fund surplus $x$, which can be  easily applied in practice.
\vskip 5pt
In reality, the impacts of different parameters on the optimal feedback functions are important, which can provide a guidance for the firm and the union to adjust their strategies. Next, we show the economic behaviours of the firm and the union explicitly and reveal the economic explanations to illustrate the rationality of our robust equilibrium strategy.
\vskip 5pt
Eq.~(\ref{pi}) for the robust investment strategy has a similar form as in \cite{Maenhout2004Robust}. $\hat{\pi}_*$ is an inverse proportional function of the volatility and $\mu+\delta$ and is proportional to the Sharpe ratio, which coincides with the Merton's line. When considering ambiguity, the firm becomes more conservative and decreases the risky allocation. By comparing Eq.~(\ref{pi}) and the results in \cite{Josa2019Equilibrium} ignoring ambiguity, the risk aversion parameter is modified to add the term of the ambiguity aversion parameter {${\mu} $.} Besides, we  {see} that the parameters of the union do not affect $\hat{\pi}_*$, i.e., the optimal feedback function of the firm only relies on its own risk aversion  and ambiguity aversion coefficients.
\vskip 5pt

The robust equilibrium  benefit $\hat{P}_*$ is given in Eq.~(\ref{p}), which includes a parameter $A$ given in Eq.~(\ref{A}). As the form of $A$ is complicated, the impacts of different parameters on $\hat{P}_*$  are not distinct. {We calculate the derivatives of the benefit ratio}
\begin{eqnarray*}
	\frac{\hat{P}_*}{x}=A^{-\frac{1}{\gamma}}=\frac{\alpha}{\gamma}-\frac{1-\gamma}{\gamma}\left(r+\left(\frac{1}{\mu+\delta}-\frac{\lambda+\gamma}{2\left(\mu+\delta\right)^2}\right)\theta^T\theta\right)
\end{eqnarray*}
with respect to the time preference of the union $\alpha$, risk-free interest rate $r$, Sharpe ratio $\theta$, the risk aversion coefficient of the firm $\delta$ and the ambiguity aversion coefficient of the firm $\mu$, the risk aversion coefficient of the union $\gamma$ and the aversion coefficient of the union  $\lambda$  explicitly to show the effects. We {see} that although  $\hat{\pi}_*$ does not depend on the union's aversion parameters, $\frac{\hat{P}_*}{x}$ is influenced by the firm's aversion parameters. The first-order derivatives are calculated as follows:
\begin{eqnarray}
	\frac{\partial A^{-\frac{1}{\gamma}}}{\partial\alpha}&=&\frac{1}{\gamma},
	\label{piandao1}\\
	\frac{\partial A^{-\frac{1}{\gamma}}}{\partial r}&=&-\frac{1-\gamma}{\gamma},
	\label{piandao2}\\
	\frac{\partial A^{-\frac{1}{\gamma}}}{\partial\theta}&=&\frac{\left(1-\gamma\right)\left[\left(\lambda+\gamma\right)-2\left(\mu+\delta\right)\right]}{\gamma\left(\mu+\delta\right)^2},
	\label{piandao3}\\
	\frac{\partial A^{-\frac{1}{\gamma}}}{\partial\delta}&=&\frac{\partial A^{-\frac{1}{\gamma}}}{\partial\mu}=\frac{\partial A^{-\frac{1}{\gamma}}}{\partial\left(\mu+\delta\right)}=\frac{\left(1-\gamma\right)\left[\left(\mu+\delta\right)-\left(\lambda+\gamma\right)\right]}{\gamma\left(\mu+\delta\right)^3}\theta^T\theta,\label{piandao4}\\
	\frac{\partial A^{-\frac{1}{\gamma}}}{\partial\gamma
	}&=&\frac{r-\alpha}{\gamma^2}+\frac{2\left(\mu+\delta\right)-\lambda-\gamma^2}{2\gamma^2\left(\mu+\delta\right)^2}\theta^T\theta,
	\label{piandao5}\\
	\frac{\partial A^{-\frac{1}{\gamma}}}{\partial\lambda
	}&=&\frac{\left(1-\gamma\right)}{2\gamma\left(\mu+\delta\right)^2}\theta^T\theta.
	\label{piandao6}
\end{eqnarray}
As the increase of the time preference makes the union more concerned with the present earnings,  Eq.~(\ref{piandao1}) shows that $\frac{\hat{P}_*}{x}$ increases with $\alpha$. In Eq.~(\ref{piandao2}), we  {see} that in the case of low risk aversion ($0<\gamma<1$), the equilibrium benefit ratio decreases while in the case of high risk aversion ($\gamma>1$), the effect is opposite. In spite that the firm is always risk aversion, the union acts completely opposite for ($0<\gamma<1$) and  ($\gamma>1$). As  the empirical studies indicate that the range of risk aversion coefficient is between 1 and 10, see \cite{Azar2006Measuring}, we only consider the union with high risk aversion  ($\gamma>1$). When the risk-free interest rate increases, the expected returns of the risk-free asset and risky assets all increase. Therefore, the fund surplus is expected to increase, {as a result,} the union can claim large benefit ratio from the fund. In our model, {$\theta$, $\delta$ and $\mu$}  influence the behaviour of the firm and {further} the fund surplus. As such, {$\theta$, $\delta$ and $\mu$ have large effects on the benefit ratio of the union.} Eq.~(\ref{piandao3}) shows that for the union with high risk aversion, $\frac{\hat{P}_*}{x}$ decreases (increases) when the Sharpe ratio $\theta$ increases in the case $\left(\lambda+\gamma\right)>2\left(\mu+\delta\right)$ $\left(\left(\lambda+\gamma\right)<2\left(\mu+\delta\right)\right)$.  {Particularly,} when $\left (\lambda+\gamma\right)=2\left(\mu+\delta\right)$,  $\frac{\hat{P}_*}{x}$ and $\theta$ are independent. $\delta$ and $\mu$ appear as a whole in Eq.~(\ref{p}). Eq.~(\ref{piandao4}) shows that the aversion coefficients $\delta$ and $\mu$ of the firm not only have the same effects on $\frac{\hat{P}_*}{x}$, but also have the same effects on $P_*$.  For the union with high risk aversion, the benefit ratio $\frac{\hat{P}_*}{x}$ decreases (increases) as $\mu+\delta$ increases  when $\left(\mu+\delta\right)>\left(\lambda+\gamma\right)\;\left(\left(\mu+\delta\right)<\left( \lambda+\gamma\right)\right)$.
\vskip 5pt
Although the firm's aversion parameters $\delta$ and $\mu$ have the same impacts on both strategies, the effects of the union's aversion parameters $\gamma$ and $\lambda$ on $\frac{\hat{P}_*}{x}$ are different. Eq.~(\ref{piandao5})  shows that the influence of $\gamma$ on $\frac{\hat{P}_*}{x}$ is complicated depending on the {relationship among $r$, $\alpha$, $\gamma$,$\mu$, $\delta$, $\lambda$ and  $\theta$.} The right hand of Eq.~(\ref{piandao5}) does not have clear economic meanings and  the benefits may increase or decrease with $\gamma$.  Eq.~(\ref{piandao6}) shows that for the  union with high risk aversion, $\frac{\hat{P}_*}{x}$ decreases as the union's ambiguity aversion coefficient $\lambda$ increases.  This is {because}  when $\lambda$ increases, the union  has less confidence about the reference model and will become more conservative  when making decisions. Thus, the union will decrease benefits from the fund. More illustrations about the relationship between $\frac{\hat{P}_*}{x}$ and the risk and ambiguity aversion parameters   are in Section \ref{NA}.

\subsubsection{\bf Pareto optimality}
In this {subsection,  in some specific case, similar to that of} \\ \cite{Josa2019Equilibrium}, we show that the robust equilibrium strategy  {is}  Pareto optimal. We consider the case when the union and the firm have the same risk aversion parameters and ambiguity aversion parameters, i.e., $\gamma=\delta$ {and} $\lambda=\mu$. Then the robust equilibrium strategy $\left(\pi_*,P_*\right)$ is {Pareto optimal, i.e.,} there is no admissible strategy $\left(\pi,P\right)\in\Lambda$ such that
\begin{eqnarray}
	\inf_{h_U\in\mathcal{H}_U\left(\pi,P\right)}J_U\left(s,x,\pi,P,h_U\right)\ge J_U\left(s,x,\pi_*,P_*,h_U^*\right),
	\label{shuangshi1}\\
	\inf_{h_F\in\mathcal{H}_F\left(\pi,P\right)}J_F\left(s,x,\pi,P,h_F\right)\ge J_F\left(s,x,\pi_*,P_*,h_F^*\right),
	\label{shuangshi2}
\end{eqnarray}
and at least one of the {last} two inequalities holds strictly.
\vskip 5pt
\cite{Josa2019Equilibrium} derive the general form of the Pareto optimal strategy and show that their equilibrium strategy is Pareto efficient. In this paper, we study a particular Pareto optimal problem in which the union and the firm cooperate to maximize the expected discount utility of the union. The robust stochastic control problem (PI0) is as follows{:}
\begin{equation*}
	V_0\left(s,x\right)\triangleq\sup_{\left(\pi,P\right)\in\Lambda}\inf_{h_0\in\mathcal{H}_U\left(\pi,P\right)}J_U\left(s,x,\pi,P,h_0\right).
\end{equation*}
Similar to Proposition \ref{pro3.1}, {the associated HJBI equation to Problem (PI0)} is
\begin{equation}
0=\sup_{\hat{\pi},\hat{P}}\inf_{\hat{h}_0}\left\{\mathcal{A}^{\hat{\pi},\hat{P},\hat{h}_0}W^0\left(s,x\right)+\Phi^U\left(s,x,\hat{\pi},\hat{P},\hat{h}_0,W^0\right)\right\}.
	\label{hjbfc3}
\end{equation}
We solve the HJBI equation (\ref{hjbfc3}) and list the results in the following theorem.
\begin{theorem}
{Assume that $A_0$ given in Eq.~(\ref{A_0}) is positive. Then} the  HJBI equation  (\ref{hjbfc3}) has a $C^{1,2}\left(\mathbf{R}_+^2\right)$ solution:
	\begin{equation}\nonumber
		W^0\left(s,x\right)=A_0e^{-\alpha s}\frac{x^{1-\gamma}}{1-\gamma}.
	\end{equation}
	Besides, {the minimum point $\hat{h}_{0*}$,  the maximum points $\hat{\pi}_{0*}$ and $\hat{P}_{0*}$} in Eq.~(\ref{hjbfc3})  are respectively
	\begin{eqnarray*}
		\hat{h}_{0*}&=&\frac{\lambda}{\gamma+\lambda}\theta,\\
		\hat{\pi}_{0*}&=&\hat{\pi_0}_*\left(x\right)=\frac{1}{\gamma+\lambda}\Sigma^{-1}\left(b-r\vec{1}\right)x,\\
		\hat{P}_{0*}&=&\hat{P_0}_*\left(x\right)=\left(A_0\right)^{-\frac{1}{\gamma}}x,\\
	\end{eqnarray*}
	where
	\begin{equation}
		A_0=\left[ \frac{\alpha}{\gamma}-\frac{1-\gamma}{\gamma}\left(r+\frac{\theta^T\theta}{2\left(\gamma+\lambda\right)}\right) \right]^{-\gamma}.
		\label{A_0}
	\end{equation}
	\label{the4.1}
\end{theorem}
\begin{proof}
	See Appendix \ref{a:the4.1}.
\end{proof}
For $\left(\pi_0^*,P_0^*\right)\in \Lambda,\; h_0^*\in\mathcal{H}_U\left(\pi_0^ *,P_0^*\right)$,
we call $\left(\pi_0^*,P_0^*\right)$ the robust optimal strategy and $h_0^*$ the worst case measure transformation process, {respectively,} if they satisfy
\begin{eqnarray*}
	V_0\left(s,x\right)=\inf_{h_0\in \mathcal{H}_U\left(\pi_0^*,P_0^*\right)}J_U\left(s,x,\pi_0^*,P_0^*,h_0\right)=J_U\left(s,x,\pi_0^*,P_0^*,h_0^*\right).
\end{eqnarray*}
Moreover, $V_0\left(s,x\right)$ is called the value function of Problem (PI0).
\begin{theorem}\label{T1}
Suppose that $A_0$ in Theorem \ref{the4.1} is positive, then  $W^0=V_0$ is the value function of Problem (PI0). {Let $\left(\pi,P\right)=\left(\{\hat{\pi}_{0*}\left(X\left(t\right)\right): t\geq s\},\{\hat{P}_{0*}\left(X\left(t\right)\right): t\geq s\}\right)$,} then SDE~(\ref{yuanfangcheng00}) has a unique solution denoted by $X^*$. If {
		 $\left(\pi_0^* ,P_0^*\right)\triangleq \left(\{\hat{\pi}_{0*}\left(X^*\left(t\right)\right): t\geq s\},\{\hat{P}_{0*}\left(X^*\left(t\right)\right): t\geq s\}\right)$} and $h_0^*$ is a constant process with value $\hat{h}_{0*}$, {then the} $\left(\pi_0^* ,P_0^*\right)$ is the robust optimal strategy for the firm and the union, {and}  $h_0^*$ is the worst case measure transformation process.
\end{theorem}
\begin{proof}
{$A_0$, $\hat{\pi}_{0*}$, $\hat{P}_{0*}$ and $\hat{h}_{0*}$ are exactly the  forms of the results obtained in Theorem \ref{the3.1} when $\gamma=\delta$ and $\lambda=\mu$. As such, based on Lemma \ref{pro3.0}, we know that $\left(\pi_0^*,P_0^*\right)$ is well-defined. In addition, based on  Lemma \ref{pro3.2}, we have $\left(\pi_0^*,P_0^*\right)\in \Lambda$ and $ h_0^*\in\mathcal{H}_U\left(\pi_0^ *,P_0^*\right)$. As such, we obtain the properties similar to Lemma \ref{pro3.3} as the procedure of the proof  is the same as that of Lemma \ref{pro3.3}. Finally, similar with the proof of Theorem \ref{the3.2}, Theorem \ref{T1} follows.}
\end{proof}
The next theorem shows that the robust equilibrium strategy in some specific case is Pareto optimal.
\begin{theorem}
{Suppose that} the union and firm have the same risk aversion parameters and ambiguity aversion parameters, i.e., {$\gamma=\delta$ and $\lambda=\mu$, and $A>0$ and $B>0$.} Then the robust equilibrium strategy $\left(\pi_*,P_*\right)$ obtained in Theorem \ref{the3.2} {is} Pareto optimal.
\end{theorem}
\begin{proof}
	{$A_0$, $\hat{\pi}_{0*}$, $\hat{P}_{0*}$ and $\hat{h}_{0*}$ obtained in Problem  (PI0) are exactly the results obtained in Theorem \ref{the3.1} when $\gamma=\delta$ and $\lambda=\mu$. As such,} the robust equilibrium strategy $\left(\pi_*,P_*\right)$ obtained in Problem (PI) is exactly the robust optimal strategy for both players to cooperate to maximize the expected discount utility of the union.  The equality in Eq.~(\ref{shuangshi1}) holds if and only if $\left(\pi,P\right)=\left(\pi_*,P_*\right)$, thus $\left(\pi_*,P_* \right)$ {is}  Pareto optimal.
\end{proof}
\vskip 5pt
The pareto optimality shows that both the union and firm  have achieved a {win-win} result of cooperation in a non-zero-sum game.
\vskip 5pt
\subsection{\bf The second game} {In this subsection}, we show the robust equilibrium strategy of Problem (PI1) in subsection \ref{sec:game2}. First, we also show the admissible sets of the strategy $(\pi,P)$, measure transformation processes $h_{\bar{U}}$ and $h_P$ to guarantee the well-posedness of Problem (PI1). Then, we present the {associated} HJBI equations to the {problem} and obtain the explicit forms of the robust equilibrium strategy as well as the worst case measure transformation processes. The optimality of the equilibrium strategy {is proved} in detail.
\subsubsection{\bf Admissible sets}
For fixed $s\ge 0,x>0$, in order to ensure the feasibility of the Girsanov's Theorem, the well-posedness of the robust problem and the optimality of the solution to the HJBI equations, we should restrict the space of the strategy {$(\pi,P)=\left\{\left(\pi\left(t\right),P\left(t\right )\right): t\geq s\right\}$} and the process {$(h_{\bar{F}},h_{\bar{U}})=\left\{\left(h_{\bar{F}}\left(t\right), h_{\bar{U}}\left(t\right)\right): t\geq s\right\}$}. The admissible strategy set of $(\pi,P)$ {is}  denoted by $\Lambda$. The firm's and union's  admissible measure transformation process sets under $\left(\pi,P\right)$ are denoted by $\mathcal{H}_{\bar{F}}\left(\pi,P\right)$ {and} $\mathcal{H}_{\bar{U}}\left(\pi,P\right)$, respectively.
\vskip 5pt
The definition of $\Lambda$ in Problem (PI1) is the same as in Problem (PI), see Definition \ref{def2.1}. The conditions of space $\Lambda$ guarantee the well-posedness of SDE~(\ref{yuanfangcheng00}). As in Problem (PI1), the union has the same objective function as in Problem (PI), the requirement of $h_{\bar{U}}$ is also the same as in Problem (PI). We say   $h_{\bar{U}}\in\mathcal{H}_{\bar{F}}\left(\pi,P\right)$ if $h_{\bar{U}}$ satisfies all conditions in Definition \ref{def2.12}. However, in Problem (PI1), the firm aims to maximize the probability of reaching the upper level before the lower level, which is quite different from Problem (PI). And the definition of $\mathcal{H}_{\bar{F}}\left(\pi,P\right)$ is  different from that {of} Definition \ref{def2.2}.
\begin{definition}
For any $\left(\pi,P\right)\in\Lambda$, denote $X$ as the unique solution of SDE~(\ref{yuanfangcheng00}) under the strategy $\left(\pi,P\right)$. 
We say that {$h_{\bar{F}}=\left \{h_{\bar{F}}\left(t\right): t\geq s\right \}$} is admissible about $\left(\pi,P\right)$, if {$\pi\left(t\right)>0$, $P\left(t\right)>0$, $\forall t\ge s$, $\mathbb{Q}^{h_{\bar{F}}}$-a.s.} and $h_{\bar{F}}$ satisfies
	\begin{enumerate}
\item[(i)]
	$
		\E_{s,x}\left[\exp\left\{\int_{s}^{T}\frac{h_{\bar{F}}\left(X\left(t\right)\right)h_{\bar{F}}\left(X\left(t\right)\right)}{2}\rd t\right\}\right]<\infty, \quad \forall T\ge s,
	$
\item[(ii)]
	$
		\mathcal{T}<\infty, \quad \mathbb{Q}^{h_{\bar{F}}}\text{-a.s..}
	$
\end{enumerate}
	All admissible $h_{\bar{F}}$ about $\left(\pi,P\right)$ form a set denoted by $\mathcal{H}_{\bar{F}}\left(\pi,P\right)$.
	\label{def5.3}
\end{definition}
In Definition \ref{def5.3}, the Novikov condition (i) guarantees that Girsanov's Theorem can be applied from the reference measure  $\mathbb{P}$ to $\mathbb{Q}^{h_{\bar{F}}}$. Condition (ii) ensures the well-posedness of the objective function $J_{\bar{F}}\left(s,x,\pi,P,h_{\bar{F}}\right)$ given in Eq.~(\ref{jf22}).


\subsubsection{\bf Robust equilibrium strategy}
 {In this subsection}, we derive the explicit forms of the robust equilibrium strategy, the corresponding worst case measure transformation process  and the value functions of Problem (PI1). {Here, stochastic dynamic programming method works for solving Problem (PI1)}. First, we  {obtain the associated} HJBI equations to Problem (PI1) in the following proposition.


\begin{proposition}[HJBI equations]
Let {$W^{\bar{U}}\left(\cdot,\cdot\right)\in C^{1,2}\left(\left(l,v\right)\times\mathbf{R}_+\right)$, $W^{\bar{F}}\left(\cdot,\cdot\right)\in C^{1,2}\left(\left(l,v\right)\times\mathbf{R}_+\right)$.
Then the associated} HJBI equations to Problem (PI1) are
	\begin{equation}
		0=\sup_{\hat{P}}\inf_{\hat{h}_{\bar{U}}}\left\{\mathcal{A}^{\hat{\pi},\hat{P},\hat{h}_{\bar{U}}}W^{\bar{U}}\left(s,x\right)+\Phi^{\bar{U}}\left(s,x,\hat{\pi},\hat{P},\hat{h}_{\bar{U}},W^{\bar{U}}\right)\right\},
		\label{xhjbfc1}
	\end{equation}
\begin{equation}
		0=\sup_{\hat{\pi}}\inf_{\hat{h}_{\bar{F}}}\left\{\mathcal{A}^{\hat{\pi},\hat{P},\hat{h}_{\bar{F}}}W^{\bar{F}}\left(s,x\right)+\Phi^{\bar{F}}\left(s,x,\hat{\pi},\hat{P},\hat{h}_{\bar{F}},W^{\bar{F}}\right)\right\},
		\label{xhjbcf2}
	\end{equation}
with boundary conditions $W^{\bar{F}}\left(s,l\right)=0$ and $ W^{\bar{F}}\left(s,v\right)=1$, where
\begin{eqnarray}
		\Phi^{\bar{U}}\left(s,x,\hat{\pi},\hat{P},\hat{h}_{\bar{U}},W^{\bar{U}}\right)&\triangleq
		&e^{-\alpha s}\frac{\hat{P}^{1-\gamma}}{1-\gamma}+\left(1-\gamma\right)
		\frac{\hat{h}_{\bar{U}}^T\hat{h}_{\bar{U}}}{2\lambda}W^{\bar{U}}\left(s,x\right),\nonumber\\
		\Phi^{\bar{F}}\left(s,x,\hat{\pi},\hat{P},\hat{h}_{\bar{F}},W^{\bar{F}}\right)&\triangleq&\frac{\hat{h}_{\bar{F}}^T\hat{h}_{\bar{F}}}{2\mu}
		\left(W^{\bar{F}}\left(s,x\right)+c\right),\nonumber
\end{eqnarray}
and $\mathcal{A}^{\hat{\pi},\hat{P},\hat{h}}$ is defined as in Eq.~(\ref{suanzi}).
	
	\label{pro5.1}
\end{proposition}
\begin{proof}
	Similar with Proposition \ref{pro3.1}, 	the derivation is simple and we omit it here, see \cite{Maenhout2006Robust} for details.
\end{proof}
{The following theorem} presents  solutions of the HJBI equations in Proposition \ref{pro5.1} as well as the maximum points and the minimum points under {some} assumptions.
\begin{theorem}
If {$\alpha>r$, $\mu\neq 1$, $0<\eta<1$, $c=\frac{l^{1-\eta}}{v^{1-\eta}-l^{1-\eta}}$ and}
	\begin{eqnarray*}
		\Delta \triangleq\left[\left(1-\frac{\gamma}{2}\right)\theta^T\theta\right]^2-2\left(\alpha-r\right)\left(1-\gamma\right)\left(\lambda+\gamma\right)\theta^T\theta>0,
	\end{eqnarray*}
the HJBI equations {(\ref{xhjbfc1})-(\ref{xhjbcf2})} in Proposition \ref{pro5.1} {have} $C^{1,2}\left(\left(l,v\right)\times\mathbf{R}_+\right)$
{solutions}:
	\begin{eqnarray*}
		W^{\bar{U}}\left(s,x\right)&=&Ee^{-\alpha s}\frac{x^{1-\gamma}}{1-\gamma},
		\\
		W^{\bar{F}}\left(s,x\right)&=&\frac{x^{1-\eta}-l^{1-\eta}}{v^{1-\eta}-l^{1-\eta}}.
	\end{eqnarray*}
	Moreover, for fixed  $\left(s,x\right)$, the maximum point $\hat{P}_*$ and the minimum point $\hat{h}_{U*}$ in Eq.~(\ref{xhjbfc1}) are given by
	\begin{eqnarray*}		\hat{P}_*&=&E^{-\frac{1}{\gamma}}x, \\
		\hat{h}_{U*}&=&\frac{\lambda}{\omega}\theta,
	\end{eqnarray*}
	and the maximum point $\hat{\pi}_*$ and the minimum point $\hat{h}_{F*}$ in Eq.~(\ref{xhjbcf2}) are given by
	\begin{eqnarray*}
				\hat{\pi}_*&=&\frac{1}{\omega}\Sigma^{-1}\left(b-r\vec{1}\right)x,\\
		\hat{h}_{F*}&=&\frac{\mu\left(1-\eta\right)}{\omega}\theta,
	\end{eqnarray*}
{where the constants $E$, $\omega$ and $\eta$} are
	\begin{eqnarray}
		E&=&\left[ r+\frac{\theta^T\theta}{2\omega} \right]^{-\gamma},\label{e}\\
		\omega&=&\frac{\left(1-\frac{\gamma}{2}\right)\theta^T\theta+\sqrt{\Delta}}{2\left(\alpha-r\right)},\label{omega}\\
		\eta&=&\frac{\omega-\mu}{1-\mu}.\label{eta}
	\end{eqnarray}
	\label{jiexhjb}
\end{theorem}
\begin{proof}
See Appendix \ref{a:jiexhjb}.
\end{proof}
\begin{remark}
Although we have introduced the constant $c$ in Eq.~(\ref{jf222}), {the exact value of $c$ is presented here} in Theorem \ref{jiexhjb}, {as such, we see that $c$ depends on $v,l$ and $\eta$.} {If  $c=0$, then we find that Problem (PI1) losses homogeneity, and the explicit form of the robust equilibrium strategy can  not be obtained.} In order to ensure the homogeneity in the solving procedure of Problem (PI1), $c$ is involved and depends on the financial system. In fact, the parameter $c$ does not affect the economic meaning of the penalty term in Eq.~(\ref{jf222}).
\end{remark}
\vskip 5pt
The results in Theorem \ref{jiexhjb} show that $\hat{h}_{U*}$ and $\hat{h}_{F*}$ are two constant processes, while $\hat{P}_*$ and $\hat{\pi}_*$ are two proportional functions  of $x$. As such, we always use $\hat{P}_*\left(x\right)$ and $\hat{\pi}_*\left(x\right)$ to represent $\hat{P}_*$ and $\hat{\pi}_*$.
\vskip 5pt
Under the strategy {$\pi=\{\hat{\pi}_*\left(X\left(t\right)\right): t\geq s\}$ and $P=\{\hat{P}_*\left(X\left(t\right)\right): t\geq s\}$, based on} Lemma \ref{lemm3.1} in the appendix, we know that SDE~(\ref{yuanfangcheng00}) has a unique solution $X^*$ under different probability measures ($\mathbb{P}$, $\mathbb{Q}^{h_{\bar{U}}^*}$ and  $\mathbb{Q}^{h_{\bar{F}}^*}$):
{\!\!\!\!\!\!\!\!\!\!\!\!\!\!\!\!{\begin{eqnarray}
\!\!\!\!\!\!\!\!X^*(t)\!\!\!\!\!&=&\!\!\!x\!\exp\left\{\left(r+\frac{\theta^T\theta}{\omega}-E^{-\frac{1}{\gamma}}-\frac{\theta^T\theta}{2\omega^2}\right)\left(t-s\right)+\frac{1}{\omega}\theta^T\left(W(t)-W(s)\right)\right\}\nonumber
		\\		\!\!\!\!\!\!\!&=&\!\!\!x\!\exp\left
\{\left(r+\frac{\theta^T\theta}{\omega}\!-\!E^{-\frac{1}{\gamma}}
\!\!-\!\!\theta^T\theta\frac{\lambda}{\omega^2}-\frac{\theta^T\theta}{2\omega^2}\right)
\left(t\!-\!s\right)+\frac{1}{\omega}\theta^T\left(W^{h_{\bar{U}}^*}(t)\!-\!W^{h_{\bar{U}}^*}(s)\right)\right\}\nonumber
		\\
		\!\!\!\!\!\!&=&\!\!\!x\!\exp\!\left\{\!\left(\!r+\!\frac{\theta^T\theta}
		{\omega}\!-\!E^{-\frac{1}{\gamma}}
\!-\!\theta^T\theta\frac{\mu\left(1\!-
\!\eta\right)}{\omega^2}\!-
\!\frac{\theta^T\theta}{2\omega^2}
\right)\left(t\!\!-\!\!s\right)\!+
\!\frac{1}{\omega}\theta^T\left(W^{h_{\bar{F}}^*}(t)\!\!-\!\!W_s^{h_{\bar{F}}^*}(s)\right)\right\}.
		\label{xzuiyoucaifu}
\end{eqnarray}}}
\vskip 5pt
In Theorem \ref{jiexhjb}, {we have obtained a pair of strategy $\left(\pi_*,P_*\right)$, where $\pi_*=\left\{\hat{\pi}_*\left(X\left(t\right)\right): t\geq s\right\}$,  $P_*=\left\{\hat{P}_*\left(X\left(t\right)\right):t\geq s\right\}$, and a pair of measure transformation process  $\left(h_{\bar{F}}^*,h_{\bar{U}}^*\right)$, where $h_{\bar{F}}^*\left(t\right)\equiv \hat{h}_{\bar{F}}, h_{\bar{U}}^*\left(t\right)\equiv \hat{h}_{\bar{U}}, t\ge s$. We need to verify that the strategy and processes are admissible. Indeed, we easily know    $ \pi_*\left(t\right)>0$, $P_*\left(t\right)>0$, $\forall t\ge s$, $\mathbb{P}\text{-a.s.}$ ($\mathbb{Q}^{h_{\bar{U}}}\text{-a.s.}, \mathbb{Q}^{h_{\bar{F}}}\text{-a.s.}$) and $\mathcal{T}<\infty,\mathbb{P}\text{-a.s.}$ ($ \mathbb{Q}^{h_{\bar{U}}}\text{-a.s.},\mathbb{Q}^{h_{\bar{F}}}\text{-a.s.}$). As such, Condition (i) in Definition \ref{def2.1} holds. Eq.~(\ref{xzuiyoucaifu}) shows the existence of the solution to SDE (\ref{yuanfangcheng00}), i.e., Condition (ii) in Definition \ref{def2.1}  also holds, thus $\left(\pi_*,P_*\right)\in\Lambda$. Similar to the proof of Lemma  \ref{pro3.2}, we can get $h_{\bar{U}}^*\in\mathcal{H}_{\bar{U}}\left(\pi_*,P_*\right)$. As $h_{\bar{F}}^*$ is a constant, $h_{\bar{F}}^*\in\mathcal{H}_{\bar{F}}\left(\pi_*,P_*\right)$.}

\vskip 5pt
{Next, we only need to prove the optimality of $(\pi_*,P_*)$ to Problem (PI1), similar to Lemma \ref{pro3.3}, we present the following lemma first.}
\begin{lemma}
	If the assumptions in Theorem \ref{jiexhjb} hold, {then} the two functions $W^{\bar{U}}\left(\cdot,\cdot\right)$ {and} $W^{\bar{F}}\left(\cdot,\cdot\right)$ presented in Theorem \ref{jiexhjb} satisfy
	\begin{enumerate}
    \item[(i)]$\forall$ $\hat{h}_{\bar{U}}\in\mathbf{R}$,$\quad\forall$ $\hat{h}_{\bar{F}}\in\mathbf{R},\quad\forall
        \left(t,y\right)\in\left(l,v\right)\times\left(0,\infty\right),$
	\begin{eqnarray*}		&&\mathcal{A}^{\hat{\pi}_*\left(y\right),\hat{P}_*\left(y\right),\hat{h}_{\bar{U}}}W^{\bar{U}}\left(t,y\right)+\Phi^{\bar{U}}\left(t,y,\hat{\pi}_*\left(y\right),\hat{P}_*\left(y\right),\hat{h}_{\bar{U}},W^{\bar{U}}\right)\ge 0,\nonumber\\		&&\mathcal{A}^{\hat{\pi}_*\left(y\right),\hat{P}_*\left(y\right),\hat{h}_{\bar{F}}}W^{\bar{F}}\left(t,y\right)+\Phi^{\bar{F}}\left(t,y,\hat{\pi}_*\left(y\right),\hat{P}_*\left(y\right),\hat{h}_{\bar{F}},W^{\bar{F}}\right)\ge 0.
\end{eqnarray*}
 \item[(ii)]$\forall $ $\hat{P}\in\mathbf{R}$,$\quad\forall$ $\hat{\pi}\in\mathbf{R}$,$\quad\forall$ $ \left(t,y\right)\in\left(l,v\right)\times\left(0,\infty\right),$
	\begin{eqnarray*}		
&&\mathcal{A}^{\hat{\pi}_*\left(y\right),
\hat{P},\hat{h}_{U*}}W^{\bar{U}}\left(t,y\right)+\Phi^{\bar{U}}\left(t,y,\hat{\pi}_*\left(y\right),\hat{P},\hat{h}_{U*},W^{\bar{U}}\right)\le 0,\nonumber\\		&&\mathcal{A}^{\hat{\pi},\hat{P}_*\left(y\right),\hat{h}_{F*}}W^{\bar{F}}\left(t,y\right)+\Phi^{\bar{F}}\left(t,y,\hat{\pi},\hat{P}_*\left(y\right),\hat{h}_{F*},W^{\bar{F}}\right)\le 0.
	\end{eqnarray*}
 \item[(iii)]$\forall \ \left(t,y\right)\in\left(l,v\right)\times\left(0,\infty\right),$
	\begin{eqnarray*}
		&&\mathcal{A}^{\hat{\pi}_*\left(y\right),\hat{P}_*\left(y\right),\hat{h}_{U*}}W^{\bar{U}}\left(t,y\right)+\Phi^{\bar{U}}\left(t,y,\hat{\pi}_*\left(y\right),\hat{P}_*\left(y\right),\hat{h}_{U*},W^{\bar{U}}\right)= 0,\nonumber\\
		&&\mathcal{A}^{\hat{\pi}_*\left(y\right),\hat{P}_*\left(y\right),\hat{h}_{F*}}W^{\bar{F}}\left(t,y\right)+\Phi^{\bar{F}}\left(t,y,\hat{\pi}_*\left(y\right),\hat{P}_*\left(y\right),\hat{h}_{F*},W^{\bar{F}}\right)= 0.\nonumber
	\end{eqnarray*} \item[(iv)]
	$\forall $\ $\left(\pi,P\right)\in\Lambda$, \ $h_{\bar{U}}\in\mathcal{H}_{\bar{U}}\left(\pi,P\right)$, $h_{\bar{F}}\in\mathcal{H}_{\bar{F}}\left(\pi,P\right)$, and  for any sequence of stopping time $\left\{\mathcal{T}_N\right\}_{N\ge 1}\uparrow \mathcal{T}$, $ \mathbb{Q}^{h_{\bar{F}}}$-a.s., the following equations hold
	\begin{eqnarray}
		&&\lim_{T\rightarrow+\infty}\E_{s,x}^{h_{\bar{U}}}W^{\bar{U}}\left(T,X\left(T\right)\right)=0,\\
		&&\lim_{N\rightarrow+\infty}\E_{s,x}^{h_{\bar{F}}}W^{\bar{F}}\left(\mathcal{T}_N,X\left(\mathcal{T}_N\right)\right)=\E_{s,x}^{h_{\bar{F}}}h\left(X\left(\mathcal{T}\right)\right),\label{ui+as}
	\end{eqnarray}
	where $X$ is the unique solution of SDE~(\ref{yuanfangcheng00}).
		\end{enumerate}
	\label{proxxingzhi}
\end{lemma}
\begin{proof}
See Appendix \ref{a:proxxingzhi}.
\end{proof}
Finally, we show that $(\pi_*,P_*)$ is indeed the robust equilibrium strategy of Problem (PI1), {i.e., we} verify that the solution of the HJBI equations in Proposition \ref{pro5.1} and the related maximum or minimum points exactly solve Problem (PI1).
\begin{theorem}
Suppose that the assumptions in Theorem \ref{jiexhjb} hold. Then $W^{\bar{U}}$ {and} $W^{\bar{F}}$ in Proposition \ref{pro5.1} are value functions of the union and firm in the second robust game. $\left(\pi_*,P_*\right)$ is the robust equilibrium strategy of Problem (PI1). {In addition,} the related worst case measure transformation processes of the firm and the union are given by {$h_{\bar{F}}^*$ and $h_{\bar{U}}^*$, respectively.} Eq.~(\ref{xzuiyoucaifu}) shows the explicit form of the corresponding optimal wealth process.
\label{xyanzheng}
\end{theorem}
\begin{proof}
See Appendix \ref{a:xyanzheng}.
\end{proof}
We call   $\left(\hat{\pi}_*\left(\cdot\right),\hat{P}_*\left(\cdot\right)\right)$ and  $\left(\hat{h}_{F*},\hat{h}_{U*}\right)$ as the optimal feedback functions and the worst case measure transformation parameters of the firm and the union, respectively. Similar with the first game, we are interested in the optimal feedback functions of the two players' strategies. We {see} that the robust equilibrium strategy is still proportional to the wealth of the fund, which can be easily applied in practice.
\vskip 5pt
\section{\bf Sensitivity analysis}
\label{NA}
In this section, we illustrate the impacts of the risk and ambiguity aversion parameters on the robust equilibrium strategies $(\hat{P}_*,\hat{\pi}_*)$ in the two games. For simplicity, we consider a financial market with one risk-free bond and one risky asset. The parameters we adopt are the same as in \cite{Josa2019Equilibrium}: in the bull market, $(b,\sigma)=(0.144604,0.10748)$ while in the bear market, $(b,\sigma)=(0.014,0.2678)$. {In addition}, $r=0.01$, $\alpha=\beta=0.02$. We restrict $\gamma$, $\delta\in\left[1,10\right]$, which means that both the union and firm have high risk aversion coefficients.
\vskip 5pt
\begin{figure}[htbp]
	\centering
	\includegraphics[scale=0.62]{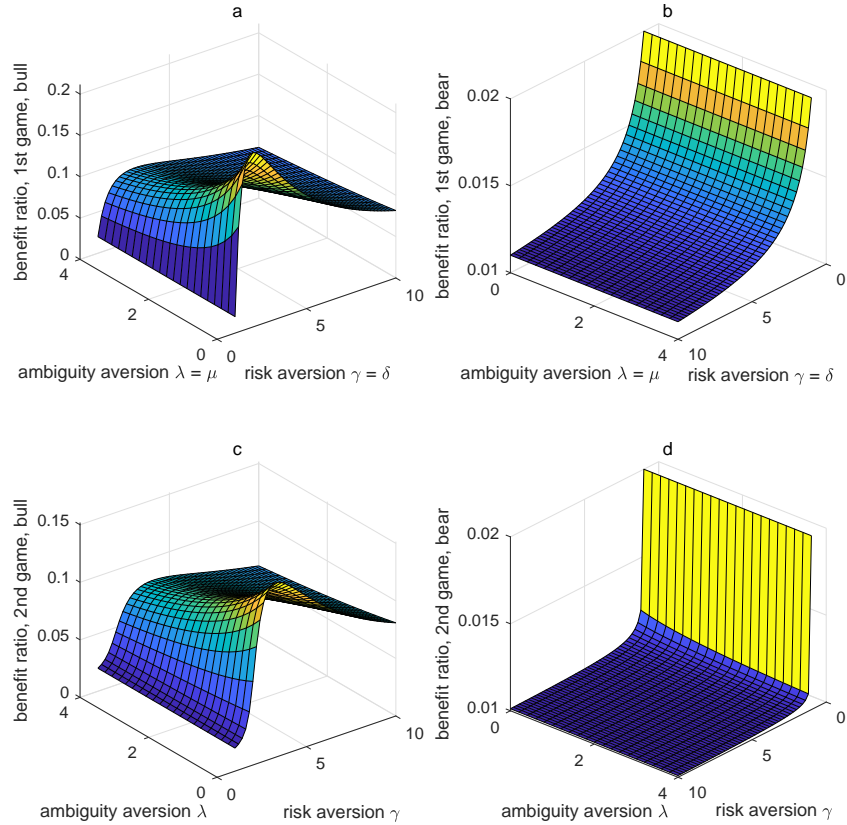}
	\caption{Benefit ratio $\frac{\hat{P}_*}{x}$ in two games, bull and bear.}
	\label{figure1}
\end{figure}

We first show the impacts in the first game, where the firm seeks to maximize the expected discounted utility of the fund surplus. The robust investment ratio $\frac{\hat{\pi}_*}{x}$ has a simple form and has been discussed in {Subsection} \ref{RA}. To compare the influences of the risk aversion parameter and ambiguity aversion parameter on the robust benefit ratio separately, we set $\gamma=\delta\in\left[1,10\right]$ and $\lambda=\mu\in\left[0,4\right]$, which means that the union and the firm have the same risk aversion coefficients and ambiguity aversion coefficients. Figure.~\ref{figure1}(a) shows that the impact of the ambiguity aversion is much more moderate than the risk aversion in the bull market. The benefit ratio increases with the risk aversion parameter rapidly first and the decreases slowly. When the firm and the union have less confidence about the financial model, the benefits claimed from the fund decrease. The impacts in the bear market are  illustrated in Figure.~\ref{figure1}(b). We  {observe}  from Figure.~\ref{figure1}(b) that  ambiguity aversion has very little influence while the risk aversion has a large impact on the benefit ratio in the bear market. The benefit ratio of the union is always a decreasing function of the risk aversion parameter. Moreover, comparing Fig.~\ref{figure1}(a) and Fig.~\ref{figure1}(b), we see that compared with in the bull market, the risky asset has less expected return and the union will be more conservative when the economy is in recession.

\vskip 5pt
{Next,} we focus on the second game, where the firm seeks to maximize  the probability of reaching an upper level before hitting a lower level. {Based on}  Theorem \ref{jiexhjb}, the robust equilibrium strategies rely on the union's aversion parameters while are independent of the firm's ambiguity aversion. Fig.~\ref{figure1}(c) and Fig.~\ref{figure1}(d) depict the evolution of the benefit ratio with the union's risk aversion parameter $\gamma$ and ambiguity aversion parameter $\lambda$  in a bull and bear market, respectively. {Comparing Figs.~\ref{figure1}(a,b) and Figs.~\ref{figure1}(c,d), we  {see} that the impacts of the ambiguity and risk aversion parameters on the benefit ratio in the two games are similar.}  The ambiguity aversion parameter always has a negative effect on the benefit ratio. However, in the bull market, the benefit ratio is an inverted $U$-shaped function of risk aversion parameter while in the bear market, the benefit ratio decreases with $\gamma$. {In addition}, the decrease in the bear market is very drastic. When $\gamma$ increases, {the benefit ratio decreases rapidly}. Observing Figs.~\ref{figure1}(a,b) and Figs.~\ref{figure1}(c,d), we {also see} that compared with the first game, the union  adopts  a  conservative behaviour in the second game.

\begin{figure}[h]
	\centering
	\includegraphics[scale=0.66]{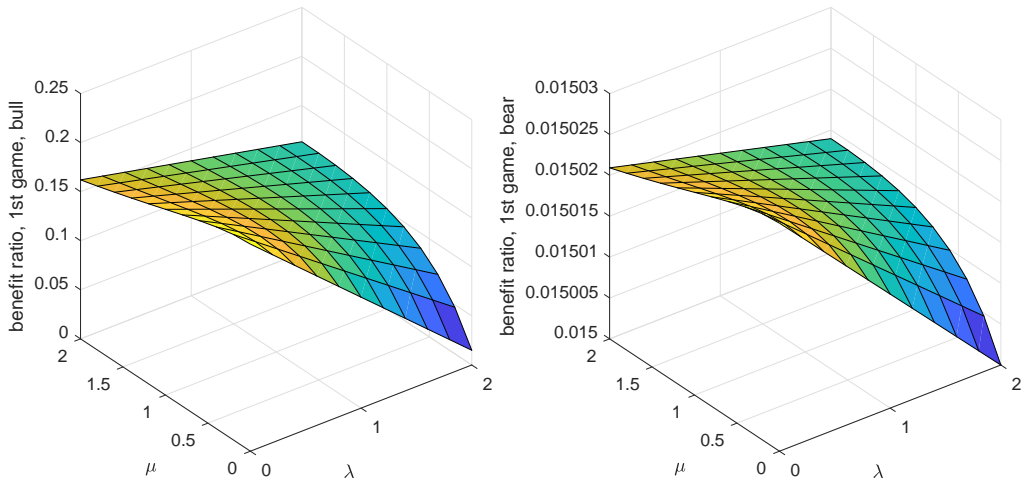}
	\caption{Benefit ratio $\frac{\hat{P}_*}{x}$ in first game, bull and bear, $\gamma=\delta=2$.}
	\label{figure2}
\end{figure}
Previously, we suppose that the firm and the union have the same aversion parameters. In this paper, we are mainly interested in the effects of model uncertainty on the players' behaviours. Next, we set $\gamma=\delta=2$ and study the impacts of the firm's and union's ambiguity aversion parameters separately. Fig.~\ref{figure2} reveals that in the first game, no matter in a bull or bear market, the benefit ratio decreases with the union's ambiguity aversion parameter $\lambda$ {while increases first and then decreases with the firm's ambiguity aversion parameter $\mu$, see Eq.~(\ref{piandao4})}. Besides, when $\mu$ increases, the benefit ratio is less sensitive to $\lambda$. It is also natural that the benefit ratio in the bear market is still lower that in the bull market.


\begin{figure}[h]
	\centering
	\includegraphics[scale=0.62]{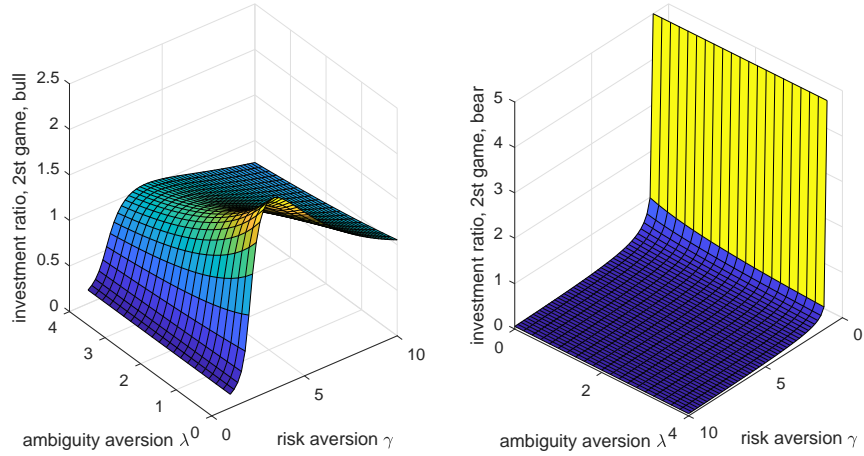}
	\caption{Proportion of fund $\frac{\hat{\pi}_*}{x}$ invested in the risky asset in the second game, bull and bear.}
	\label{figure3}
\end{figure}
\vskip 5pt
Fig.~\ref{figure3} shows that the evolution of  investment ratio has a similar pattern with the benefit ratio in the second game shown in Fig.~\ref{figure1}(c) and Fig.~\ref{figure1}(d). The impact of the risk aversion parameter $\gamma$ on the investment strategy is larger than the ambiguity aversion parameter $\lambda$. When $\gamma$ increases, the investment ratio increases rapidly to the highest point and then  decreases slowly in the bull market. Meanwhile, the investment ratio is a decreasing function of $\lambda$. However, in the bear market, the effect of $\gamma$ is much larger than that of $\lambda$.  When $\gamma$ increases, the investment strategy decreases rapidly to around zero.
\vskip 5pt
In fact, these two types of aversion parameters of the union always indirectly affect the firm's investment strategy $\frac{\hat{\pi}_*}{x}$ through the benefit ratio $\frac{\hat{P}_*}{x}$. When the benefits claimed by the union increase, the firm will increase risky allocation to achieve high returns to ensure the stability of the fund surplus. As such, the evolution of the investment strategy is similar with the benefit ratio in the second game.

%
\vskip 5pt
Comparing the two sub-figures in Fig.~\ref{figure3}, we find that when $\gamma$ is low, the investment ratio in the bear market is higher than that in the bull market. However, the conclusion is opposite when the risk aversion coefficient $\gamma$ is relatively high. When the economy is in recession and the union has low risk aversion, the benefit ratio is high relatively to the risk-free interest rate. As such, the fund surplus may always be above the lower level. In order to reach the upper level, the firm will adopt a  much more aggressive investment strategy.  On the other hand, if the union has high risk aversion, {then} the benefit ratio becomes close to $r$, {as such}, the investment strategy decreases.

%
\vskip 15pt
\section{\bf Conclusion}
In this paper, we investigate the overfunded DB pension plan game proposed in \\ \cite{Josa2019Equilibrium}. The firm manages the fund surplus and invests in the financial market while the pension fund participants claim part of the fund surplus as additional benefits. Different from \cite{Josa2019Equilibrium},  {we} assume that both the firm and the union are ambiguity aversion and consider two robust stochastic non-zero-sum games for them. The aims of the firm and the union are to searching the robust equilibrium strategies under worst case scenario.
\vskip 5pt
In this first game, both players maximize their expected discounted utility under worst case scenario. However, in the second game, the firm aims to minimize the probability of the fund surplus reaching a low level. {Using} stochastic dynamic programming method, we obtain the explicit forms of the robust equilibrium strategies, the value functions and worst case measures. {Particularly,}  in the first game, when the firm and the union have the same aversion parameters, the robust equilibrium strategy is Pareto optimal. We {also} see that the robust equilibrium strategies are proportional to the fund surplus in these two games. Besides, the robust investment strategy of the firm in the first game is only related to the risk aversion and ambiguity aversion of the firm. However, in the second game, the robust investment strategy  only depends on the risk aversion and ambiguity aversion of the union. Finally, we present sensitivity analysis to {reveal} the player's economic behaviors. Ambiguity aversion always prompts the union to adopt a more conservative strategy whether the economy is booming or sluggish, but this effect is smaller when the economy is in recession.

\vskip 5pt
On the one hand, the robust optimal strategy in our work can provide a guidance to the firm and the union in practice to manage risk of model uncertainty. On the other hand, different from many previous work, we replace the uniform integrability condition by another condition, which is useful in {proving} the verification theorem of robust control problem.
\vskip 15pt
{\bf Acknowledgements.}
The authors acknowledge the support from the National Natural Science Foundation of China (Grant  No.11901574, No.11871036, and No.11471183). The authors also thank the members of the group of {Mathematical Finance and Actuarial Science}  at the Department of Mathematical Sciences, Tsinghua University for their feedbacks and useful conversations.
\vskip 15pt
\appendix
\renewcommand{\theequation}{\thesection.\arabic{equation}}
\section{Proof of Theorem \ref{the3.1}.}\label{a:the3.1}
\begin{proof}
{Based on Remarks \ref{rem2} and \ref{rem3}}, we guess that {the solutions of the HJBI equations (\ref{hjbfc1})-(\ref{hjbcf2}) have the forms:}
	\begin{equation}
		W^U\left(s,x\right)=e^{-\alpha s}W^U\left(x\right),
		W^F\left(s,x\right)=e^{-\beta s}W^F\left(x\right),
		\label{xingzhi1}
	\end{equation}
with the properties
	\begin{equation}
		\left(1-\gamma\right)W^U\left(x\right)>0,\quad
		\left(1-\delta\right)W^F\left(x\right)>0.
		\label{xingzhi2}
	\end{equation}
{By using} Eq.~(\ref{xingzhi1}), the HJBI {equations are rewritten as follows:}
	\begin{eqnarray}
		0&=&\sup_{\hat{P}}\inf_{\hat{h}_U}\{-\alpha W^U\left(x\right)+\left(rx+\hat{\pi}^T\left(b-r\vec{1}\right)-\hat{P}-\hat{\pi}^T\sigma \hat{h}_U\right)W^U_x\left(x\right)\nonumber
		\\
		&&+\frac{1}{2}\hat{\pi}^T\Sigma\hat{\pi} W^U_{xx}\left(x\right)+\frac{\hat{P}^{1-\gamma}}{1-\gamma}+\left(1-\gamma\right)\frac{\hat{h}_U^T\hat{h}_U}{2\lambda}W^U\left(x\right)\},
		\label{jianshi1}\\
		0&=&\sup_{\hat{\pi}}\inf_{\hat{h}_F}\{-\beta W^F\left(x\right)+\left(rx+\hat{\pi}^T\left(b-r\vec{1}\right)-\hat{P}-\hat{\pi}^T\sigma \hat{h}_F\right)W^F_x\left(x\right)\nonumber\\
		&&+\frac{1}{2}\hat{\pi}^T\Sigma\hat{\pi} W^F_{xx}\left(x\right)+\frac{x^{1-\delta}}{1-\delta}+\left(1-\delta\right)\frac{\hat{h}_F^T\hat{h}_F}{2\mu}W^F\left(x\right)\}.
		\label{jianshi2}
	\end{eqnarray}
{Using Eq.~(\ref{xingzhi2}), we know that the function on the right side of Eq.~(\ref{jianshi1}) is an open-up quadratic function of $\hat{h}_U$. As such, the first-order condition is}
	\begin{equation}
		\hat{h}_{U*}=\frac{\lambda\sigma^T\hat{\pi} W_x^U\left(x\right)}{\left(1-\gamma\right)W^U\left(x\right)}.
		\label{hu00}
	\end{equation}
Denote the right side of Eq.~(\ref{jianshi1}) by $f_1\left(\hat{P}\right)$, then we have
	$f_1^{'}=\hat{P}^{-\gamma}-W_x^U\left(x\right)$ and $f_1^{''}=-\gamma \hat{P}^{-1-\gamma}<0$, as such,
	\begin{equation}
		\hat{P}_*=\left(W_x^U\left(x\right)\right)^{-\frac{1}{\gamma}}.
		\label{p0}
	\end{equation}
    Similarly, if Eq.~(\ref{xingzhi2}) holds, when
	\begin{equation}
		\hat{h}_{F*}=\frac{\mu\sigma^T\hat{\pi} W_x^F\left(x\right)}{\left(1-\delta\right)W^F\left(x\right)},
		\label{hf00}
	\end{equation}
{the right side of Eq.~(\ref{jianshi2}), i.e., $f_2\left(\hat{\pi}\right)$} attains the minimum and  {is}
	\begin{eqnarray*}
		f_2\left(\hat{\pi}\right)&=&
		-\beta W^F\left(x\right)+\left(rx+\hat{\pi}^T\left(b-r\vec{1}\right)-\hat{P}\right)W^F_x\left(x\right)\nonumber
		\\
		&&+\frac{1}{2}\hat{\pi}^T\Sigma\hat{\pi} W^F_{xx}\left(x\right)
		+\frac{x^{1-\delta}}{1-\delta}-\frac{\mu\hat{\pi}^T\Sigma\hat{\pi} \left(W_x^F\left(x\right)\right)^2}{2\left(1-\delta\right)W^F\left(x\right)}.
	\end{eqnarray*}
{As such,}
	\begin{equation*}
		\nabla f_2=\left(b-r\vec{1}\right)W_x^F\left(x\right)-\left[\frac{\mu \left(W_x^F\left(x\right)\right)^2}{\left(1-\delta\right)W^F\left(x\right)}-W_{xx}^F\left(x\right)\right]\Sigma\hat{\pi}.
	\end{equation*}
	If
	\begin{equation}
		D\triangleq\frac{\mu \left(W_x^F\left(x\right)\right)^2}{\left(1-\delta\right)W^F\left(x\right)}-W_{xx}^F\left(x\right)>0,
		\label{xingzhi3}
	\end{equation}
	{then}
	\begin{equation*}
		\nabla^2 f_2=-D\Sigma\prec 0.
	\end{equation*}
{Thus},
	\begin{equation}
		\hat{\pi}_*=\frac{\Sigma^{-1}\left(b-r\vec{1}\right)W_x^F\left(x\right)}{D}.
		\label{pi0}
	\end{equation}
Substituting Eq.~(\ref{pi0}) into Eqs.~(\ref{hu00}) and (\ref{hf00}) {yields}
	\begin{eqnarray}
		\hat{h}_{U*}&=&\frac{\lambda\theta W_x^U\left(x\right) W_x^F\left(x\right)}{\left(1-\gamma\right)W^U\left(x\right)D},
		\label{hu0}\\
		\hat{h}_{F*}&=&\frac{\mu\theta \left(W_x^F\left(x\right)\right)^2}{\left(1-\delta\right)W^F\left(x\right)D}.
		\label{hf0}
	\end{eqnarray}
Substituting Eqs.~(\ref{p0}), (\ref{pi0}), (\ref{hu0}) and (\ref{hf0}) into Eqs.~(\ref{jianshi1}) and (\ref{jianshi2}), {we obtain}
	\begin{eqnarray}
		0&=&-\alpha W^U\left(x\right)+\frac{\gamma}{1-\gamma}\left(W_x^U\left(x\right)\right)^{1-\frac{1}{\gamma}}+rxW_x^U\left(x\right)\nonumber
		\\&&+\theta^T\theta\frac{W_x^U\left(x\right)W_x^F\left(x\right)}{D}+\frac{1}{2}\theta^T\theta\frac{\left(W_x^F\left(x\right)\right)^2}{D^2}W_{xx}^U\left(x\right)\nonumber
		\\&&-\theta^T\theta\frac{\lambda}{2\left(1-\gamma\right)W^U\left(x\right)}\left(\frac{W_x^U\left(x\right)W_x^F\left(x\right)}{D}\right)^2,
		\label{fangc1}
		\\
		0&=&-\beta W^F\left(x\right)+\frac{x^{1-\delta}}{1-\delta}+rxW_x^F\left(x\right)
		+\theta^T\theta\frac{\left(W_x^UF\left(x\right)\right)^2}{D}\nonumber
		\\&&-\left(W_x^U\left(x\right)\right)^{-\frac{1}{\gamma}}W_x^F\left(x\right)+\frac{1}{2}\theta^T\theta\frac{\left(W_x^F\left(x\right)\right)^2}{D^2}W_{xx}^F\left(x\right)\nonumber
		\\&&-\theta^T\theta\frac{\mu}{2\left(1-\delta\right)W^F\left(x\right)}\left(\frac{\left(W_x^F\left(x\right)^2\right)}{D}\right)^2.
		\label{fangc2}
	\end{eqnarray}
Suppose  $W^U\left(x\right)=A\frac{x^{1-\gamma}}{1-\gamma}$, $W^F\left(x\right)=B\frac{x^{1-\delta}}{1-\delta}$, where $A$ and $B$ are some  constants. In this case, Eq.~(\ref{xingzhi2}) is equivalent to  {$A>0$ and $B>0$.} {As $D=\left(\mu+\delta\right)Bx^{-1-\delta}>0$ and Eq.~(\ref{xingzhi3}) holds,}  Eqs.~(\ref{fangc1}) and (\ref{fangc2}) {are rewritten as follows:}
	
	\begin{eqnarray*}
		0&=&-\frac{\alpha}{1-\gamma}x^{1-\gamma}A+\frac{\gamma}{1-\gamma}x^{1-\gamma}A^{1-\frac{1}{\gamma}}+rx^{1-\gamma}A\nonumber
		\\&&+\frac{\theta^T\theta}{\mu+\delta}x^{1-\gamma}A-\theta^T\theta\frac{\gamma}{2\left(\mu+\delta\right)^2}x^{1-\gamma}A-\theta^T\theta\frac{\lambda}{2\left(\mu+\delta\right)^2}x^{1-\gamma}A,
		\\
		0&=&-\frac{\beta}{1-\delta}x^{1-\delta}B+\frac{1}{1-\delta}x^{1-\delta}+rx^{1-\delta}B+\frac{\theta^T\theta}{\mu+\delta}x^{1-\delta}B\nonumber
		\\&&-A^{-\frac{1}{\gamma}}x^{1-\delta}B-\theta^T\theta\frac{\delta}{2\left(\mu+\delta\right)^2}x^{1-\delta}B-\theta^T\theta\frac{\mu}{2\left(\mu+\delta\right)^2}x^{1-\delta}B,
	\end{eqnarray*}
  which are  Eqs.~(\ref{A}) and (\ref{B}) in  Theorem \ref{the3.1}.
\vskip 5pt	
{Therefore, Eqs.~(\ref{jieu})  and (\ref{jief}) are the $C^{1,2}$ solutions of the HJBI equations. Substituting the expressions $W^U\left(x\right)=A\frac{x^{1-\gamma}}{1-\gamma}$ and $ W^F\left(x\right)=B\frac{x^{1-\delta}}{1-\delta}$ into Eqs.~(\ref{p0}), (\ref{hu0}), (\ref{pi0}) and (\ref{hf0}), we obtain Eqs.~(\ref{p}), (\ref{hu}), (\ref{pi}) and  (\ref{hf}), correspondingly.}
\end{proof}
\vskip 5pt
\section{Proof of Lemma \ref{pro3.0}.}\label{a:pro3.0}

First, we present the following lemma about the geometric Brownian motion.

\begin{lemma}
	Let {$\left\{B\left(t\right): t\geq 0\right\}$} be an $n$-dimensional standard Brownian motion, $c$ and $m$ are constants, and $v\in\mathbf{R}^n$.
	 If the process {$\left\{N\left(t\right): t\geq 0\right\}$} satisfies the stochastic differential equation $\rd N\left(t\right)=cN\left(t\right)\rd t+N\left(t\right)v^T\rd B(t)$,
	then $N\left(t\right)=N\left(0\right)\exp\left\{\left(c-\frac{1}{2}v^Tv\right)t+v^TB\left(t\right)\right\}$. Moreover, if $N\left(0\right)$ and $B\left(t\right)$ are independent, then {$\E N\left(t\right)^m=\E\left\{N\left(0\right)^m\exp\left\{m\left(c-\frac{1}{2}v^Tv\right)t+\frac{1} {2}m^2v^Tvt\right\}\right\}$.}
	\label{lemm3.1}
\end{lemma}
\begin{proof}
As {$\{N\left(t\right): t\geq 0\}$} is a geometric Brownian motion, the explicit form of $N\left(t\right)$ can be easily derived by applying It\^{o}'s formula to $\ln N\left(t\right)$.	{Because $\left\{\exp\left\{mv^TB\left(t\right)-\frac{1}{2}m^2v^Tvt\right\}: t\geq 0\right\}
	$} is a martingale,  we have
	\begin{eqnarray*}
		\E N\left(t\right)^m&=&\E N\left(0\right)^m\exp\left\{m\left(c-\frac{1}{2}v^Tv\right)t\right\}+\E \exp\left\{mv^TB\left(t\right)\right\}\nonumber
		\\&=&\E N\left(0\right)^m\exp\left\{m\left(c-\frac{1}{2}v^Tv\right)t
		+\frac{1}{2}m^2v^Tvt\right\}.
	\end{eqnarray*}
\end{proof}
{Next,}  we prove Lemma \ref{pro3.0} {as follows:}
\begin{proof}
{As	 $X$ satisfies the following SDE:}
\begin{eqnarray*}		\rd X\left(t\right)&=&X\left(t\right)
\left(r+\frac{\theta^T\theta}{\mu+\delta}
-A^{-\frac{1}{\gamma}}\right)\rd t+X\left(t\right)
\frac{1}{\mu+\delta}\theta^T\rd W\left(t\right)\nonumber\\
&=&X\left(t\right)\left(r+\frac{\theta^T\theta}{\mu+\delta}
-A^{-\frac{1}{\gamma}}-\theta^T\theta
\frac{\lambda}{\left(\mu+\delta\right)^2}\right)\rd t
+X\left(t\right)\frac{1}{\mu+\delta}\theta^T\rd W^{h_U^*}\left(t\right)\nonumber\\
&=&X\left(t\right)\left(r+\frac{\theta^T\theta}{\mu+
\delta}-A^{-\frac{1}{\gamma}}-\theta^T\theta
\frac{\mu}{\left(\mu+\delta\right)^2}\right)\rd t+X\left(t\right)
\frac{1}{\mu+\delta}\theta^T\rd W^{h_F^*}\left(t\right).\nonumber\\
\end{eqnarray*}
{Based on Lemma \ref{lemm3.1},} we have
	{\small
{\begin{eqnarray}
X^*\left(t\right)\!\!\!&=&\!\!\!\!\!x\!\exp\left\{\left(r+\frac{\theta^T\theta}{\mu+\delta}-A^{-\frac{1}{\gamma}}-\frac{\theta^T\theta}{2\left(\mu+\delta\right)^2}\right)\left(t-s\right)+\frac{1}{\mu+\delta}\theta^T\left(W\left(t\right)-W\left(s\right)\right)\right\}\nonumber
			\nonumber\\ \!\!\!&=&\!\!\!x\!\exp\left\{\left(r\!+
\!\frac{\theta^T\theta}{\mu+\delta}\!-
\!A^{-\frac{1}{\gamma}}\!-\!\theta^T\theta\frac{\lambda}{\left(\mu\!+
\!\delta\right)^2}\!-\!\frac{\theta^T\theta}{2\!\left(\mu\!+
\!\delta\right)^2}\!\right)\left(t\!\!-\!\!s\right)\!+\!\frac{1}{\mu\!+\!\delta}\theta^T\!\left(\!W^{h_U^*}\left(t\right)\!-\!W^{h_U^*}\left(s\right)\!\right)\right\}\nonumber
			\nonumber\\
			\!\!\!\!\!\!&=&\!\!\!\!x\!\exp\!\left\{\left(\!r\!+\!
\!\frac{\theta^T\theta}{\mu\!+\!\delta}\!-\!
\!A^{-\frac{1}{\gamma}}\!-\!\theta^T
\theta\frac{\mu}{\left(\mu\!+\!\delta\right)^2}
\!-\!\frac{\theta^T\theta}{2\left(\mu\!+\!\delta\right)^2}\right)\!
\left(\!t\!\!-\!\!s\right)\!+\!\frac{1}{\mu\!+\!
\delta}\theta^T\left(W^{h_F^*}\left(t\right)\!-\!W^{h_F^*}\left(s\right)\!\right)\!\!\right\}.
			\label{biaoda}
	\end{eqnarray}}}
\end{proof}
\section{Proof of Lemma \ref{pro3.2}.}\label{a:pro3.2}
\vskip 5pt
\begin{proof}
As $\left(\hat{\pi}_*,\hat{P}_*\right)$ is continuous, the process $\left(\pi_*,P_*\right)=\left(\hat{\pi}_*\left(X^*\right),\hat{P}_*\left(X^*\right)\right)$ is adapted to the filtration $\{\mathcal{F}_t\}_{t\ge s}$. Under the strategy $\left(\pi_*,P_*\right)$,  SDE~(\ref{yuanfangcheng00}) has a pathwise unique solutions $X^*$. {Based on} the expression of $X^*$ given by Eq.~(\ref{biaoda}), we
  {know  $X^*\left(t\right)>0,\ \forall \ t\ge s,\mathbb{P}\text{-a.s.} \left(\mathbb{Q}^{h_U^*}\text{-a.s.},\mathbb{Q}^{h_F^*}\text{-a.s.}\right)$.} As such, we have $\pi_*\left(t\right)>0$, $P_*\left(t\right)> 0$, $\forall$ $ t\ge s$, $\mathbb{P}\text{-a.s.}$  ($\mathbb{Q}^{h_U^*}\text{-a.s.},\mathbb{Q}^{h_F^*}\text{-a.s.}$). Therefore $\left(\pi_*,P_*\right)\in\Lambda$. Noting that {$h_U^*$ and $h_F^*$} are constant processes, Condition (i) in Definition \ref{def2.2} {holds. Based on}  Lemma \ref{lemm3.1}, we have
	\begin{eqnarray*}		{\E_{s,x}^{h_U^*}\left\{\left(X\left(t\right)\right)^{1-\gamma}\right\}}&=&x^{1-\gamma}\;\exp\left\{\left(1-\gamma\right)\left(r+\frac{\theta^T\theta}{\mu+\delta}-A^{-\frac{1}{\gamma}}-\theta^T\theta\frac{\lambda+\gamma}{\left(\mu+\delta\right)^2}\right)\left(t-s\right)\right\},
		\\		{\E_{s,x}^{h_F^*}\left\{\left(X\left(t\right)\right)^{1-\delta}\right\}}&=&x^{1-\delta}\;\exp\left\{\left(1-\delta\right)\left(r+\frac{\theta^T\theta}{2\left(\mu+\delta\right)}-A^{-\frac{1}{\gamma}}\right)\left(t-s\right)\right\}.
	\end{eqnarray*}
As such, Conditions (ii) and (iii) in Definition \ref{def2.2} {hold if}	
	\begin{eqnarray*}
&&\left(1-\gamma\right)\left(r+\frac{\theta^T\theta}{\mu+\delta}-A^{-\frac{1}{\gamma}}-\theta^T\theta\frac{\lambda+\gamma}{\left(\mu+\delta\right)^2}\right)<\alpha,\\
		&&\left(1-\delta\right)\left(r+\frac{\theta^T\theta}{2\left(\mu+\delta\right)}-A^{-\frac{1}{\gamma}}\right)<\beta.
	\end{eqnarray*}
{The last two Ineqs. are  equivalent to $A>0$ and $ B>0$.} In addition, Condition (iv) in Definition \ref{def2.2} is equivalent to
	\begin{eqnarray*}
		\E_{s,x}^{h_U^*}\left[\int_{s}^{T}\left(X\left(t\right)\right)^{2-2\gamma}\rd t\right]<\infty, \quad \forall T\ge s,
		\\ \E_{s,x}^{h_F^*}\left[\int_{s}^{T}\left(X\left(t\right)\right)^{2-2\delta}\rd t\right]<\infty, \quad \forall T\ge s.
	\end{eqnarray*}
{Based on} Lemma \ref{lemm3.1}, {$\E_{s,x}^{h_U^*}\left\{X\left(t\right)\right\}^{2-2\gamma}$  and $\E_{s,x}^{h_F^*}\left\{X\left(t\right)\right\}^{2-2\delta}$} can be written in the form of $C_1e^{C_2\left(t-s\right)}$, {as such}, Condition (iv)  holds.
{Thus, $h_U^*\in\mathcal{H}_U\left(\pi_*,P_*\right)$ and $h_F^*\in\mathcal{H}_F\left(\pi_*,P_*\right)$ follow.}
\end{proof}
\vskip 5pt
\section{Proof of {Lemma} \ref{pro3.3}.}\label{a:pro3.3}
\begin{proof}
 {Based on} the forms of HJBI equations $\left(\ref{hjbfc1}\right)$ and $\left(\ref{hjbcf2}\right)$ with variable $\left(s,x\right)\in\mathbb{\mathbf{R}}_+^2$, {it is easy to see that
	 the properties (i) and(iii)  hold.} Property (iv) is exactly equivalent to Condition (iii) in Definition \ref{def2.2} with variable {$\left(s,x\right)\in\mathbf{R}_+^2$}, which is naturally satisfied.
%
\vskip 5pt
As the function involved in $\hat{P}$ does not depend on  $\hat{\pi}$ and $\hat{h}_U$, the first formula of Property (ii) {is} easily verified by Eq.~(\ref{hjbfc1}).
Next, we verify the second one. Let $f_3\left(\hat{\pi}\right)\triangleq\mathcal{A}^{\hat{\pi},\hat{P}_*,\hat{h}_{F*}}W_F\left(t,y\right)+\Phi^F\left(t,y,\hat{\pi},\hat{P}_*,\hat{h}_{F*},W^F\right)$. {Then we obtain the derivatives of $f_3$:}
	\begin{eqnarray*}
		\nabla f_3&=&e^{-\alpha t}\left(\left(b-r\vec{1}\right)W_x^F\left(y\right)-\frac{\mu}{\mu+\delta}\sigma\theta-\Sigma\hat{\pi} W_{xx}^F\left(y\right)\right)\nonumber\\
		&=&e^{-\alpha t}\delta By^{-\delta}\left(\frac{\left(b-r\vec{1}\right)}{\mu+\delta}-y^{-1}\Sigma\hat{\pi}\right),\\
		\nabla^2 f_3&=&-e^{-\alpha t}\delta By^{-\delta-1}\Sigma\prec 0.
	\end{eqnarray*}
 {As such, if} $\hat{\pi}=\frac{1}{\mu+\delta}\Sigma^{-1}\left(b-r\vec{1}\right)y=\hat{\pi}_*\left(y\right)$, {then} $\nabla f_3=0$.
	{Thus} , $f_3\left(\hat{\pi}\right)\le f_3\left(\hat{\pi}_*\left(y\right)\right)=0$.
\end{proof}
\vskip 5pt
\section{Proof of Theorem \ref{the4.1}.}\label{a:the4.1}
\begin{proof}
Suppose that the solution of {the HJBI equation (\ref{hjbfc3})} has properties
	\begin{equation}
		W^0\left(s,x\right)=e^{-\alpha s}W^0\left(x\right),\quad \left(1-\gamma\right)W^0\left(x\right)>0.
		\label{xingzhix0}
	\end{equation}
	Then the HJBI equation {is rewritten as follows:}
	\begin{eqnarray}
		0&=&\sup_{\hat{\pi},\hat{P}}\inf_{\hat{h}_0}\left\{-\alpha W^0\left(x\right)+\left(rx+\hat{\pi}^T\left(b-r\vec{1}\right)-\hat{P}-\hat{\pi}^T\sigma \hat{h}_0\right)W^0_x\left(x\right)\right.\nonumber
		\\
		&&+\left.\frac{1}{2}\hat{\pi}^T\Sigma\hat{\pi} W^0_{xx}\left(x\right)+\frac{\hat{P}^{1-\gamma}}{1-\gamma}+\left(1-\gamma\right)\frac{\hat{h}_0^T\hat{h}_0}{2\lambda}W^0\left(x\right)\right\}.
		\label{jianshi3}
	\end{eqnarray}
{Following} the proof of Theorem \ref{the3.1}, we have
	\begin{eqnarray}
		\hat{h}_{0*}&=&\frac{\lambda\sigma^T\hat{\pi} W_x^0\left(x\right)}{\left(1-\gamma\right)W^0\left(x\right)},
		\label{hu00x}\\
		\hat{P_0}_*&=&\left(W_x^0\left(x\right)\right)^{-\frac{1}{\gamma}}.
		\label{p0x}
	\end{eqnarray}
{As such}, the right side of Eq.~(\ref{jianshi3}) {is}
	\begin{eqnarray*}
		f_0\left(\hat{\pi}\right)&=&
		-\alpha W^0\left(x\right)+\left(rx+\hat{\pi}^T\left(b-r\vec{1}\right)-\hat{P}_0^*\right)W^0_x\left(x\right)\nonumber
		\\&&+\frac{1}{2}\hat{\pi}^T\Sigma\hat{\pi} W^0_{xx}\left(x\right)
		+\frac{\left(\hat{P_0}_*\right)^{1-\gamma}}{1-\gamma}-\frac{\lambda\hat{\pi}^T\Sigma\hat{\pi} \left(W_x^0\left(x\right)\right)^2}{2\left(1-\gamma\right)W^0\left(x\right)}.
	\end{eqnarray*}
	Similar to that of Theorem \ref{the3.1}, we obtain that if
	\begin{equation}
		D_0\triangleq\frac{\lambda \left(W_x^0\left(x\right)\right)^2}{\left(1-\gamma\right)W^0\left(x\right)}-W_{xx}^0\left(x\right)>0,
		\label{xingzhix}
	\end{equation}
	then
	\begin{equation}
		\hat{\pi_0}_*=\frac{\Sigma^{-1}\left(b-r\vec{1}\right)W_x^0\left(x\right)}{D_0}.
		\label{pi0x}
	\end{equation}
	As such, the HJBI equation becomes
	\begin{eqnarray}
		0&=&-\alpha W^0\left(x\right)+\frac{\gamma}{1-\gamma}\left(W_x^0\left(x\right)\right)^{1-\frac{1}{\gamma}}+rxW_x^0\left(x\right)\nonumber
		\\&&+\theta^T\theta\frac{\left(W_x^0\left(x\right)\right)^2}{D_0}+\frac{1}{2}\theta^T\theta\frac{\left(W_x^0\left(x\right)\right)^2}{D_0^2}W_{xx}^0\left(x\right)
		-\theta^T\theta\frac{\lambda\left(W^0\left(x\right)\right)^4}{2\left(1-\gamma\right)W^0\left(x\right)D_0^2}.
		\label{fangcx}
	\end{eqnarray}
{Solving Eq.~(\ref{fangcx}) yields} $W^0\left(x\right)=A_0\frac{x^{1-\gamma}}{1-\gamma}$, where $A_0$ is given by Eq.~(\ref{A_0}). {Noting that} Assumption (\ref{xingzhix0}) is equivalent to $A_0>0$, Assumption (\ref{xingzhix}) naturally holds. As such,	
{$\left(\hat{\pi_0}_*,\hat{P_0}_*\right)$ and $ \hat{h}_{0*}$ are derived. Thus, the proof is completed.}
\end{proof}
\vskip 5pt

\section{Proof of Theorem \ref{jiexhjb}.}\label{a:jiexhjb}
\begin{proof}
We guess that {the solutions of the HJBI equations (\ref{xhjbfc1})-(\ref{xhjbcf2}) have the forms:}
	\begin{equation}
		W^{\bar{U}}\left(s,x\right)=e^{-\alpha s}W^{\bar{U}}\left(x\right),\ \ \
		W^{\bar{F}}\left(s,x\right)=W^{\bar{F}}\left(x\right).
		\label{xxingzhi1}
	\end{equation}
Suppose that $W^{\bar{U}}$ {and} $W^{\bar{F}}$ satisfy
	\begin{equation}
		\left(1-\gamma\right)W^{\bar{U}}\left(x\right)>0,\quad
		W^{\bar{F}}\left(x\right)>0.
		\label{xxingzhi2}
	\end{equation}
If Eqs.~(\ref{xxingzhi1}) and (\ref{xxingzhi2}) hold, {then} the HJBI equation{s are} equivalent to
	\begin{eqnarray}
		0&=&\sup_{\hat{P}}\inf_{\hat{h}_{\bar{U}}}\left\{-\alpha W^{\bar{U}}\left(x\right)+\left(rx+\hat{\pi}^T\left(b-r\vec{1}\right)-\hat{P}-\hat{\pi}^T\sigma \hat{h}_{\bar{U}}\right)W^{\bar{U}}_x\left(x\right)\right.\nonumber
		\\
		&&+\left.\frac{1}{2}\hat{\pi}^T\Sigma\hat{\pi} W^{\bar{U}}_{xx}\left(x\right)+\frac{\hat{P}^{1-\gamma}}{1-\gamma}+\left(1-\gamma\right)\frac{\hat{h}_{\bar{U}}^T\hat{h}_{\bar{U}}}{2\lambda}W^{\bar{U}}\left(x\right)\right\},
		\label{xjianshi1}\\
		0&=&\sup_{\hat{\pi}}\inf_{\hat{h}_{\bar{F}}}\left\{\left(rx+\hat{\pi}^T\left(b-r\vec{1}\right)-\hat{P}-\hat{\pi}^T\sigma \hat{h}_{\bar{F}}\right)W^{\bar{F}}_x\left(x\right)\nonumber\right.\\
		&&+\left.\frac{1}{2}\hat{\pi}^T\Sigma\hat{\pi} W^{\bar{F}}_{xx}\left(x\right)+\frac{\hat{h}_{\bar{F}}^T\hat{h}_{\bar{F}}}{2\mu}\left(W^{\bar{F}}\left(x\right)+c\right)\right\}.
		\label{xjianshi2}
	\end{eqnarray}
Similar {to that} of Theorem \ref{the3.1}, we know that the minimum point $\hat{h}_{U*}$ and the maximum $\hat{P}_*$ are
	\begin{equation}
		\hat{h}_{U*}=\frac{\lambda\sigma^T\hat{\pi} W_x^{\bar{U}}\left(x\right)}{\left(1-\gamma\right)W^{\bar{U}}\left(x\right)},
		\label{xhu00}
	\end{equation}
	\begin{equation}
		\hat{P}_*=\left(W_x^{\bar{U}}\left(x\right)\right)^{-\frac{1}{\gamma}}.
		\label{xp0}
	\end{equation}
	Combining Eqs.~(\ref{xjianshi2})  and (\ref{xxingzhi2}), we have that when
	\begin{equation}
		\hat{h}_{F*}=\frac{\mu\sigma^T\hat{\pi} W_x^{\bar{F}}\left(x\right)}{W^{\bar{F}}\left(x\right)+c},
		\label{xhf00}
	\end{equation}
	 the right side of the HJBI equation attains the minimum and becomes
	\begin{eqnarray*}
		f_2\left(\hat{\pi}\right)&=&
		\left(rx+\hat{\pi}^T\left(b-r\vec{1}\right)-\hat{P}\right)W^{\bar{F}}_x\left(x\right)\nonumber
		\\
		&&+\frac{1}{2}\hat{\pi}^T\Sigma\hat{\pi} W^{\bar{F}}_{xx}\left(x\right)
		-\frac{\mu\hat{\pi}^T\Sigma\hat{\pi} \left(W_x^{\bar{F}}\left(x\right)\right)^2}{2\left(W^{\bar{F}}\left(x\right)+c\right)},
	\end{eqnarray*}
	with derivative
	\begin{equation*}
		\nabla f_2=\left(b-r\vec{1}\right)W_x^{\bar{F}}\left(x\right)-\left[\frac{\mu \left(W_x^{\bar{F}}\left(x\right)\right)^2}{W^{\bar{F}}\left(x\right)+c}-W_{xx}^{\bar{F}}\left(x\right)\right]\Sigma\hat{\pi}.
	\end{equation*}
	If we {assume}
	\begin{equation}
		F\triangleq\frac{\mu \left(W_x^{\bar{F}}\left(x\right)\right)^2}{W^{\bar{F}}\left(x\right)+c}-W_{xx}^{\bar{F}}\left(x\right)>0,
		\label{xxingzhi3}
	\end{equation}
	{then}
	\begin{equation*}
		\nabla^2 f_2=-F\Sigma\prec 0.
	\end{equation*}
	As such,
	\begin{equation}
		\hat{\pi}_*=\frac{\Sigma^{-1}\left(b-r\vec{1}\right)W_x^{\bar{F}}\left(x\right)}{F}.
		\label{xpi0}
	\end{equation}
	Substituting Eq.~(\ref{xpi0}) into Eqs.~(\ref{xhu00}) and (\ref{xhf00}) {yields}
	\begin{eqnarray}
		\hat{h}_{U*}&=&\frac{\lambda\theta W_x^{\bar{U}}\left(x\right) W_x^{\bar{F}}\left(x\right)}{\left(1-\gamma\right)W^{\bar{U}}\left(x\right)F},
		\label{xhu0}\\
		\hat{h}_{F*}&=&\frac{\mu\theta \left(W_x^{\bar{F}}\left(x\right)\right)^2}{\left(W^{\bar{F}}\left(x\right)+c\right)F}.
		\label{xhf0}
	\end{eqnarray}
	Substituting Eqs.~(\ref{xhu0}), (\ref{xp0}), (\ref{xhf0}) and (\ref{xpi0}) into Eqs.~(\ref{xjianshi1}) and (\ref{xjianshi2}), {we obtain}
\begin{eqnarray}
0&=&-\alpha W^{\bar{U}}\left(x\right)+\frac{\gamma}{1-\gamma}\left(W_x^{\bar{U}}\left(x\right)
\right)^{1-\frac{1}{\gamma}}+rxW_x^{\bar{U}}\left(x\right)\nonumber	\\
&&+\theta^T\theta\frac{W_x^{\bar{U}}\left(x\right)W_x^{\bar{F}}\left(x\right)}{F}
+\frac{1}{2}\theta^T\theta\frac{\left(W_x^{\bar{F}}\left(x\right)\right)^2}{F^2}W_{xx}^{\bar{U}}
\left(x\right)\nonumber
		\\&&-\theta^T\theta
\frac{\lambda}{2\left(1-\gamma\right)W^{\bar{U}}
\left(x\right)}\left(\frac{W_x^{\bar{U}}\left(x\right)W_x^{\bar{F}}\left(x\right)}{F}\right)^2,\\\label{xfangc0}
0&=&rxW_x^{\bar{F}}\left(x\right)
		+\theta^T\theta\frac{\left(W_x^{\bar{U}}F\left(x\right)\right)^2}{F}\nonumber
		\\&&-\left(W_x^{\bar{U}}\left(x\right)\right)^{-\frac{1}{\gamma}}W_x^{\bar{F}}\left(x\right)+\frac{1}{2}\theta^T\theta\frac{\left(W_x^{\bar{F}}\left(x\right)\right)^2}{F^2}W_{xx}^{\bar{F}}\left(x\right)\nonumber
		\\&&-\theta^T\theta
\frac{\mu}{2\left(W^{\bar{F}}\left(x\right)+c\right)}
\left(\frac{\left(W_x^{\bar{F}}\left(x\right)^2\right)}{F}\right)^2.
		\label{xfangc}
	\end{eqnarray}
{Finally, we only need to show that Eqs.~(\ref{xxingzhi2}) and (\ref{xxingzhi3}) hold, and $W^{\bar{U}}\left(x\right)=E\frac{x^{1-\gamma}}{1-\gamma}$ and $ W^{\bar{F}}\left(x\right)=\frac{x^{1-\eta}-l^{1-\eta}}{v^{1-\eta}-l^{1-\eta}}$ indeed solve Eqs~(\ref{xfangc0})-(\ref{xfangc}). The first part of Eq.~(\ref{xxingzhi2}) is equivalent to $E>0$ and the second  part naturally holds. Using $\alpha>r$, $\eta<1$ and Eq.~(\ref{omega}), we  know   $\mu\left(1-\eta\right)+\eta=\omega>0$, then $F=\left(1-\eta\right)\left[\mu\left(1-\eta\right)
+\eta\right]\frac{x^{-1-\eta}}{v^{1-\eta}-l^{1-\eta}}>0$. As such, Eq.~(\ref{xxingzhi3}) holds, and Eqs.~(\ref{xfangc0})-(\ref{xfangc}) become}
	
%
%
	\begin{eqnarray*}
		0&=&-\frac{\alpha}{1-\gamma}x^{1-\gamma}E+\frac{\gamma}{1-\gamma}x^{1-\gamma}E^{1-\frac{1}{\gamma}}+rx^{1-\gamma}E+\frac{\theta^T\theta}{\mu\left(1-\eta\right)+\eta}x^{1-\gamma}E\nonumber\\
		&&-\theta^T\theta\frac{\gamma}{2\left[\mu\left(1-\eta\right)+\eta\right]^2}x^{1-\gamma}E-\theta^T\theta\frac{\lambda}{2\left[\mu\left(1-\eta\right)+\eta\right]^2}x^{1-\gamma}E,
		\\
		0&=&r\left(1-\eta\right)\frac{x^{1-\eta}}{v^{1-\eta}-l^{1-\eta}}+\theta^T\theta\frac{1-\eta}{\mu\left(1-\eta\right)+\eta}\frac{x^{1-\eta}}{v^{1-\eta}-l^{1-\eta}}
		-E^{-\frac{1}{\gamma}}\left(1-\eta\right)\frac{x^{1-\eta}}{v^{1-\eta}-l^{1-\eta}}\nonumber\\
		&&-\theta^T\theta\frac{\eta\left(1-\eta\right)}{2\left[\mu\left(1-\eta\right)+\eta\right]^2}\frac{x^{1-\eta}}{v^{1-\eta}-l^{1-\eta}}-\theta^T\theta\frac{\mu}{2\left[\mu\left(1-\eta\right)+\eta\right]^2}\frac{x^{1-\eta}}{v^{1-\eta}-l^{1-\eta}}.
	\end{eqnarray*}
{Simplifying the last two equations,}
	\begin{eqnarray}		0&=&-\frac{\alpha}{1-\gamma}+r+\frac{\gamma}{1-\gamma}E^{-\frac{1}{\gamma}}+\frac{\theta^T\theta}{\mu\left(1-\eta\right)+\eta}-\theta^T\theta\frac{\lambda+\gamma}{2\left[\mu\left(1-\eta\right)+\eta\right]^2}\label{fc1},\\
		0&=&r+\frac{\theta^T\theta}{2 \left[\mu\left(1-\eta\right)+\eta\right]}-E^{-\frac{1}{\gamma}}\label{fc2}.
	\end{eqnarray}
	{Eq.~(\ref{eta}) implies  $\omega=\mu\left(1-\eta\right)+\eta$. Therefore, Eq.~(\ref{fc2}) holds because of Eq.~(\ref{omega}). Using Eqs.~(\ref{e}) and (\ref{omega}), Eq.~(\ref{fc1})  holds.}	
\end{proof}
\vskip 5pt
\section{Proof of {Lemma} \ref{proxxingzhi}.}\label{a:proxxingzhi}
\begin{proof}
 {Following} the proof of Lemma \ref{pro3.3}, {the} properties (i), (ii), (iii) and the first part of (iv) can be proved. {For} the proof of the second part of property (ii), the condition $\eta>0$ is required. {As such}, we only need to prove the second part of property (iv).
 \vskip 5pt	
Noting $X\left(\mathcal{T}_N\right)\in \left(l,v\right),\forall N$, we have {$\sup \limits_ N \E_{s,x}^{h_{\bar{F}}}|X\left(\mathcal{T}_N\right)|^2<\infty$,} which implies that\\ {$\left\{X\left(\mathcal{T}_N\right): N\geq 0\right\}$} is uniformly integrable. {As $\mathcal{T}<\infty, \mathbb{Q}^{h_{\bar{F}}}\text{-a.s.}$, and $W^{\bar{F}}\left(\mathcal{T}_N,X\left(\mathcal{T}_N\right)\right)\rightarrow W^{\bar{F}}\left(\mathcal{T},X\left(\mathcal{T}\right)\right), \mathbb{Q}^{h_{\bar{F}}}\text{-a.s.}$, we know that  Eq.~(\ref{ui+as}) holds.}
\end{proof}
\vskip 5pt
\section{Proof of Theorem \ref{xyanzheng}.}\label{a:xyanzheng}
\begin{proof}
{$\forall\ \left(\pi,P\right)\in\Lambda,\  h_{\bar{U}}\in\mathcal{H}_{\bar{U}}(\pi,P)$,} let $X$ be the unique solution of SDE~(\ref{yuanfangcheng00}). Similar to the proof of Theorem \ref{the3.2}, we have
\begin{eqnarray}
		\inf_{h_{\bar{U}}\in \mathcal{H}_{\bar{U}}\left(\pi_*,P_*\right)}
J_{\bar{U}}\left(s,x,\pi_*,P_*,h_{\bar{U}}\right)
		&=&\sup_{P:\left(\pi_*,P\right)\in\Lambda}\inf_{h_{\bar{U}}\in \mathcal{H}_{\bar{U}}\left(\pi_*,P\right)}J_{\bar{U}}\left(s,x,\pi_*,P,h_{\bar{U}}\right) \nonumber
		\\&=&V_{\bar{U}}^{\pi_*}\left(s,x\right)=W^{\bar{U}}\left(s,x\right)= J_{\bar{U}}\left(s,x,\pi_*,P_*,h_{\bar{U}}^*\right).\nonumber
	\end{eqnarray}
{Thus,} we only need to {show}
\begin{eqnarray}
\inf_{h_{\bar{F}}\in \mathcal{H}_{\bar{F}}\left(\pi_*,P_*\right)}
J_{\bar{F}}\left(s,x,\pi_*,P_*,h_{\bar{F}}\right)
		&=&\sup_{\pi:\left(\pi,P_*\right)\in\Lambda}\inf_{h_{\bar{F}}\in \mathcal{H}_{\bar{F}}\left(\pi,P_*\right)}J_{\bar{F}}\left(s,x,\pi,P_*,h_{\bar{F}}\right) \nonumber
		\\&=&V_{\bar{F}}^{P_*}\left(s,x\right)=W^{\bar{F}}\left(s,x\right)= J_{\bar{F}}\left(s,x,\pi_*,P_*,h_{\bar{F}}^*\right).\nonumber
	\end{eqnarray}
Using  It\^{o}'s formula, we have
	\begin{eqnarray}
		W^{\bar{F}}\left(\mathcal{T},X\left(\mathcal{T}\right)\right)=\!W^{\bar{F}}\left(s,X\left(s\right)\right)\!+\!
		\int_{s}^{\mathcal{T}}
		\mathcal{A}^{\pi,P,h_{\bar{F}}}W^{\bar{F}}\left(t,X\left(t\right)\right)\rd t\!+\!\int_{s}^{\mathcal{T}}W_x^{\bar{F}}\left(t,X\left(t\right)\right)\pi\left(t\right)^T\sigma\!\rd W^{h_{\bar{F}}}\left(t\right).\nonumber
	\end{eqnarray}
Let $\mathcal{T}_N\triangleq\mathcal{T}\wedge N \wedge \inf\left\{t>s;\int_{s}^t|W_x^{\bar{F}}\left(z,X\left(z\right)\right)\pi\left(z\right)^T\sigma|^2\rd z\ge N\right\}$,
{then $ \mathcal{T}_N  $ is a stopping time and}  $\left\{\int_{s}^{T\wedge \mathcal{T}_N}W_x^{\bar{F}}\left(t,X\left(t\right)\right)\pi\left(t\right)^T\sigma \rd W^{h_{\bar{F}}}\left(t\right); s \le T \le N \right\}$ is a $\mathbb{Q}^{h_{\bar{F}}}$ martingale. As such,
	\begin{equation}
		\E_{s,x}^{h_{\bar{F}}}W^{\bar{F}}\left(\mathcal{T}_N,X\left({\mathcal{T}_N}\right)\right)=
		W^{\bar{F}}\left(s,x\right)+\E_{s,x}^{h_{\bar{F}}}
		\int_{s}^{\mathcal{T}_N}
		\mathcal{A}^{\pi\left(t\right),P
			\left(t\right),h_{\bar{F}}\left(t\right)}W^{\bar{F}}\left(t,X\left(t\right)\right)\rd s.
		\label{xth36}
	\end{equation}
{Based on Condition (i)  in Proposition \ref{proxxingzhi}, $\pi_*\left(t\right)=\hat{\pi}_*\left(X\left(t\right)\right)$ and  $P_*\left(t\right)=\hat{P}_*\left(X\left(t\right)\right)$, we have}
	$\forall\  h_{\bar{F}}\in \mathcal{H}_{\bar{F}}\left(\pi_*,P_*\right)$,
	\begin{equation}
		W^{\bar{F}}\left(s,x\right)\le \E_{s,x}^{h_{\bar{F}}}\int_{s}^{\mathcal{T}_N}\Phi^{\bar{F}}\left(t,X\left(t\right),\pi_*\left(t\right),P_*\left(t\right),h_{\bar{F}}\left(t\right),W^{\bar{F}}\right)\rd t+\E_{s,x}^{h_{\bar{F}}}W^{\bar{F}}\left(\mathcal{T}_N,X\left({\mathcal{T}_N}\right)\right).\nonumber
	\end{equation}
Letting $N\rightarrow+\infty$, {using Condition (iv) in Proposition \ref{proxxingzhi} and non-negativity of $\Phi^{\bar{F}}$, based on integral expansion theorem, we have}	
	\begin{equation}
		W^{\bar{F}}\left(s,x\right)\le J_{\bar{F}}\left(s,x,\pi_*,P_*,h_{\bar{F}}\right).
		\label{xbudengsji1}
	\end{equation}
On the other hand, applying Condition (ii) in Proposition \ref{proxxingzhi} to Eq.~(\ref{xth36}), we {have} that  for %
$P_*,h_{\bar{F}}^*$, $\forall\  \pi$ such that $\left(\pi,P_*\right)\in\Lambda\wedge h_{\bar{F}}^*\in\mathcal{H}_{\bar{U}}\left(\pi,P_*\right)$,
	\begin{equation}
		W^{\bar{F}}\left(s,x\right)\ge \E_{s,x}^{h_{\bar{F}}^*}\int_{s}^{\mathcal{T}_N}\Phi^{\bar{F}}\left(t,X\left(t\right),\pi\left(t\right),P_*\left(t\right),h_{\bar{F}}^*\left(t\right),W^{\bar{F}}\right)\rd t+\E_{s,x}^{h_{\bar{F}}^*}W^{\bar{F}}\left(\mathcal{T}_N,X\left({\mathcal{T}_N}\right)\right).\nonumber
	\end{equation}
	Letting $N\rightarrow+\infty$,
	\begin{equation}
		W^{\bar{F}}\left(s,x\right)\ge J_{\bar{F}}\left(s,x,\pi,P_*,h_{\bar{F}}^*\right).
		\label{xbudengshi2}
	\end{equation}
Applying Condition (iii) in Proposition \ref{proxxingzhi} to Eq.~(\ref{xth36}) and letting $N\rightarrow+\infty$, we have
%
	\begin{equation}\nonumber
		W^{\bar{F}}\left(s,x\right)= J_{\bar{F}}\left(s,x,\pi_*,P_*,h_{\bar{F}}^*\right).
	\end{equation}
{Because the} rest part of the proof is the same as in the proof of Theorem \ref{the3.2},  {we} omit it here.	{Thus, the proof follows.}
\end{proof}
\bibliographystyle{apalike}
\bibliography{wpref}

\end{document}